\documentclass[reprint, nofootinbib, aps, floatfix]{revtex4-1}

\usepackage[utf8]{inputenc}
\usepackage{graphicx}
\usepackage{subfigure}
\usepackage[english]{babel}
\usepackage{dcolumn}
\usepackage{bm}
\usepackage{enumerate}
\usepackage{amsmath,amssymb,amsbsy,amsthm}
\usepackage{upgreek}
\usepackage{braket}
\usepackage{IEEEtrantools}
\usepackage{hhline}
\usepackage{multirow}
\usepackage{soul}
\usepackage{color}
\usepackage[dvipsnames]{xcolor}
\usepackage{pifont} 
\usepackage{siunitx}
\usepackage{hyperref}

\newcommand{\EQ}{EQ}
\newcommand{\ST}{ST}
\newcommand{\FLAT}{FL}

\newcommand{\p}{\partial}
\DeclareMathOperator{\GAP}{GAP}
  
\begin{document}

\title{Generalization capabilities of translationally equivariant neural networks}

\author{Srinath Bulusu}
\email{sbulusu@hep.itp.tuwien.ac.at}
\author{Matteo Favoni}
\email{favoni@hep.itp.tuwien.ac.at}
\author{Andreas Ipp}
\email{ipp@hep.itp.tuwien.ac.at}
\author{David I.~M\"uller}
\email{dmueller@hep.itp.tuwien.ac.at}
\author{Daniel Schuh}
\email[Corresponding author: ]{schuh@hep.itp.tuwien.ac.at}
\affiliation{Institute for Theoretical Physics, TU Wien, Austria}

\date{\today}

\begin{abstract}
	The rising adoption of machine learning in high energy physics and lattice field theory necessitates the re-evaluation of common methods that are widely used in computer vision, which, when applied to problems in physics, can lead to significant drawbacks in terms of performance and generalizability. One particular example for this is the use of neural network architectures that do not reflect the underlying symmetries of the given physical problem. In this work, we focus on complex scalar field theory on a two-dimensional lattice and investigate the benefits of using group equivariant convolutional neural network architectures based on the translation group. For a meaningful comparison, we conduct a systematic search for equivariant and non-equivariant neural network architectures and apply them to various regression and classification tasks. We demonstrate that in most of these tasks our best equivariant architectures can perform and generalize significantly better than their non-equivariant counterparts, which applies not only to physical parameters beyond those represented in the training set, but also to different lattice sizes.
\end{abstract}

\maketitle

\section{Introduction\label{sec:introduction}}

Machine learning has become an increasingly popular tool for a diverse range of applications in physics. Particularly suitable for the analysis of spatially arranged data are Convolutional Neural Networks (CNNs). Modern CNN architectures are based on the idea that a network's prediction should not change when the input is shifted. They rely on two key ingredients that have already been introduced by the Neocognitron~\cite{Fukushima:1980} over 40 years ago: convolutional layers (S-cells) and pooling (subsampling, downsampling) layers (C-cells). This incorporation of a translational symmetry was an essential advantage over its predecessor, the Cognitron~\cite{Fukushima:1975}. However, equivariance under translations is not guaranteed in a generic CNN, even though it is the idea behind weight sharing in the convolutional layers.

In the past decade, the computer vision community has created many different machine learning algorithms and continues to refine them. During the ImageNet Large Scale Visual Recognition Challenge (ILSVRC)~\cite{Russakovsky:2015}, which was a popular competition that was held annually from 2010 until 2017, the performance of CNNs steadily increased, and in 2012, AlexNet~\cite{Krizhevsky:2017} was the first CNN to win the classification task. However, its first convolutional layer already breaks translational equivariance by using a stride of four, as do three max pooling layers with a stride of two that are part of the network. Additionally, the output of the last convolutional layer is flattened before it is passed to the dense layers of the network. \mbox{LeNet-5}~\cite{Lecun:1998}, a very early CNN, uses a stride of one in the convolutional layers, but the average pooling layers with a stride of two break translational symmetry. An important step towards a translationally equivariant network architecture has been made by the introduction of global pooling layers. Global Average Pooling (GAP) was first introduced in~\cite{Lin:2014}, and the first winning network of the ILSVRC's classification task that makes use of it is ResNet~\cite{He:2015} from 2015's competition.

The grand success of machine learning in many different tasks has also garnered attention within other research communities. Although many ingredients can be carried over from computer vision, differences in the tasks may require a different treatment. A lot of effort has been made to incorporate global \cite{Cohen:2016, Cohen:2016a, Worrall:2017, Worrall:2018, Ecker:2018, Veeling:2018, Cohen:2019, Lafarge:2020, Pang:2020, Nicoli:2021, Scaife:2021} and gauge \cite{Kanwar:2020, Boyda:2020, Favoni:2020} symmetries in the network architecture, since they play a central role in modern physics, among other fields. Nevertheless, the most basic one, translational symmetry, is often not strictly enforced despite the fact that the task would allow for it. Oftentimes, the data are flattened somewhere in the network, as e.g.~in~\cite{Zhou:2019, Wetzel:2017, Bachtis:2020, Bachtis:2020a, Bluecher:2020, Padavala:2021}, and sometimes a convolutional or pooling operation with a stride greater than one spoils symmetry under translations, even though a global pooling layer constitutes the transition from the convolutional part of the network to its dense part, e.g.~in \cite{Wang:2021}. There are examples that make explicit use of this symmetry though, such as \cite{zhang:2021}, and we want to raise awareness that one should take it into account when choosing a network architecture.

In this paper, we focus on translational symmetry in CNNs. In addition to providing theoretical reasons for choosing a translationally equivariant architecture, we conduct experiments with three types of architecture on three different machine learning tasks. We try to find well-performing architectures for each of these types by doing an extensive search with the optimization framework \textit{optuna}~\cite{Akiba:2019}. We furthermore investigate the generalization capabilities of the three network types in two different ways. First, we examine how models generalize to different sets of physical parameters at a fixed lattice size, and, second, we inquire how well they generalize to other lattice sizes. The latter is not possible for networks including a flattening step, since they require a fixed input size. Therefore, these types of model would have to be retrained for each new lattice size.

This paper is structured as follows: we first discuss translational symmetry in section~\ref{sec:translational_symmetry} and show under which circumstances CNNs are indeed respecting equivariance under translations, as well as certain pitfalls that break the symmetry, in section~\ref{sec:SymmetryProperties}. The next three sections are devoted to three different machine learning tasks pertaining to a complex scalar field on a lattice: a regression task with the aim of predicting two observables of said scalar field (section~\ref{sec:Regression}), a classification task in which the algorithm should judge whether or not the flux of a given lattice configuration is conserved (section~\ref{sec:Classification}) and another regression task, in which the network is supposed to figure out how many flux violations are present on the lattice (section~\ref{sec:MultipleWorms}). Section~\ref{sec:ConclusionsAndOutlook} contains our conclusions and possible future research avenues. The appendices encompass information about the complex scalar field (appendix~\ref{app:ComplexScalarField}), our datasets (appendix~\ref{app:Datasets}), some supplemental proofs for section~\ref{sec:SymmetryProperties} (appendix~\ref{app:Proofs}) and some additional analysis pertaining to the regression task in section~\ref{sec:Regression} (appendix~\ref{app:PartiallyObscuredInput}).

The code and data sets used in this work are published in a separate repository\footnote{See \href{https://gitlab.com/openpixi/scalar_ml}{https://gitlab.com/openpixi/scalar\_ml}}. 

\section{Translational symmetry} \label{sec:translational_symmetry}

In this section, we exemplify symmetry aspects on a complex scalar field and explain how they may impact the choice of machine learning models. The action of a complex scalar field~$\phi$ in an external potential~$V$ in $D$~dimensions can be written as
\begin{equation}
    S = \int \! \mathrm{d}^D x \, \mathcal{L}, \label{generic_action_of_a_complex_scalar_field}
\end{equation}
with the Lagrangian density
\begin{equation}
    \mathcal{L} = \p_{\mu} \phi^* \p^{\mu} \phi - V(\phi^* \phi).
\end{equation}
The latter is covariant under translations \mbox{$x^{\mu} \rightarrow x^{\prime \mu} = x^{\mu} + a^{\mu}$}, with a constant vector $a^\mu$. This can be seen by noting that the fields transform via~\mbox{$\phi'(x^{\mu}) = \phi(x^{\mu} - a^{\mu})$} and that the partial derivative is not influenced by translations~\mbox{$\p_{\mu}' = \p_{\mu}$}. The action given by Eq.~\eqref{generic_action_of_a_complex_scalar_field} is then invariant under translations. This has important implications for the resulting physical theory because finite continuous symmetries of the action lead to conserved quantities, according to Noether's first theorem \cite{Noether:1918}. The invariance under temporal translations entails energy conservation; the invariance under spatial translations leads to momentum conservation.

\begin{figure}
    \centering
    \subfigure[~Equivariant architecture (EQ)]{\includegraphics{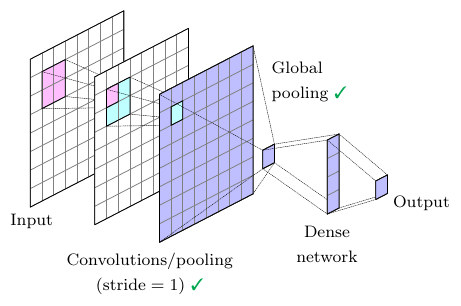}}
    \subfigure[~Strided architecture (ST)]{\includegraphics{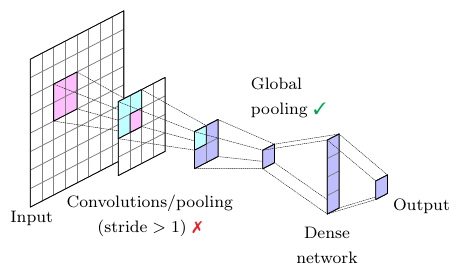}}
    \subfigure[~Flattening architecture (FL)]{\includegraphics{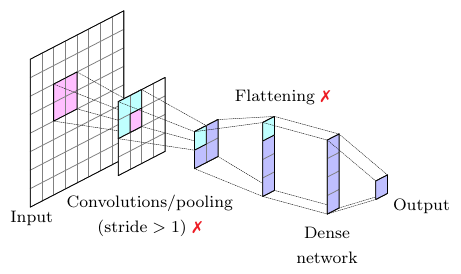}}
    \caption{
        The three different architecture types used in this study. The checkmark (\textcolor{Green}{\ding{51}}) or cross (\textcolor{Red}{\ding{55}}) indicate spatial operations in which translational symmetry is respected or violated, respectively. Translational symmetry can be violated by convolutional or pooling layers with a stride greater than one or by a flattening layer. A global pooling layer allows for the application of the same network to different lattice sizes. Each of the layers can have a number of channels (not depicted) without affecting the translational symmetry properties.}
    \label{fig:architectures}
\end{figure}

Another important symmetry of the action in Eq.~\eqref{generic_action_of_a_complex_scalar_field} is a global $U(1)$ symmetry, given by \mbox{$\phi \rightarrow e^{i \alpha} \phi$}. It implies the existence of a conserved four-current $j^\mu$ and allows the definition of a chemical potential $\mu$. The action can be modified to directly include the chemical potential via
\begin{equation}
    S = \int \! \mathrm{d}x_0\, \mathrm{d}^{D-1} \mspace{-2mu} x \mspace{2mu} \left( \lvert D_0 \phi \rvert^2 - \lvert \partial_i \phi \rvert^2 - V(\lvert \phi \rvert) \right), \label{action_complex_scalar_field_continuum_in_D_dimensions}
\end{equation}
with $D_0 = \partial_0 - i \mu$.

In the following, we consider a complex scalar field in \mbox{$1 + 1$} dimensions in a quartic potential
\begin{equation}
    V(\lvert \phi \rvert) = m^2 \lvert \phi \rvert^2 + \lambda \lvert \phi \rvert^4
\end{equation}
on the lattice with periodic boundary conditions. The parameters of the potential are the mass~$m$ and the coupling constant~$\lambda$. A discretized version of the action in Eq.~\eqref{action_complex_scalar_field_continuum_in_D_dimensions} retains its invariance under discrete translations. We then switch to a dual representation, called flux representation. It describes the same physical content as the original representation, but the variables are four integer fields (link variables) $k_{x, \nu}$ and~$l_{x, \nu}$, with~$\nu = 1,2$, instead of the complex scalar field. The corresponding partition function reads
\begin{IEEEeqnarray}{rCl}
    \nonumber Z &=& \sum_{ \{ k, l \} } \left( \prod_{x, \nu} \frac{1}{(\lvert k_{x, \nu} \rvert + l_{x, \nu})! l_{x, \nu}!} \right) \left( \prod_x e^{\mu k_{x,2}} W(f_x) \right) \\
    && \> \times \left( \prod_x \delta \left( \sum_\nu (k_{x, \nu} - k_{x-\hat{\nu}, \nu}) \right) \right), \IEEEeqnarraynumspace \label{trans:flux_representation}
\end{IEEEeqnarray}
where the outer sum is a shorthand for
\begin{equation}
    \sum_{ \{ k, l \} } = \prod_{x, \nu} \mspace{8mu}\sum_{k_{x, \nu} = - \infty}^\infty \mspace{8mu} \sum_{l_{x, \nu} = 0}^\infty.
\end{equation}
The function~$W(f_x)$ is given by
\begin{equation}
    W(f_x)=\int_0^\infty \mathrm{d}x\, x^{f_x+1}\mathrm{e}^{-\eta x^2-\lambda x^4}, \label{trans:W}
\end{equation}
and the integer field $f_x$ is defined as
\begin{equation}
f_x=\sum_\nu[|k_{x,\nu}|+|k_{x-\hat{\nu},\nu}|+2(l_{x,\nu}+l_{x-\hat{\nu},\nu})]. \label{trans:f_x_definition}
\end{equation}
The dual formulation incorporates the same symmetry properties as the
original formulation. A more detailed explanation of this procedure is given in appendix~\ref{app:ComplexScalarField} and in~\cite{Gattringer:2013b}. To ensure the flux conservation demanded by the Kronecker $\delta$ symbol in Eq.~\eqref{trans:flux_representation}, the worm algorithm~\cite{Prokofev:2001} has been employed to update the link variables~$k_{x, \nu}$. It is a local algorithm that updates contiguous field values on the lattice in successive steps. The resulting structures are known as worms. When the head of a worm meets its tail, a worm is closed; otherwise, it is open. Details about the generation of our datasets can be found in appendix~\ref{app:Datasets}.

This dual representation in two dimensions allows for a strong analogy with two-dimensional images. While every pixel of an image is described by one (grayscale) or three (color) numbers, every position of the dual lattice is described by four values. An important difference between them is their boundary conditions. Typically, image applications employ fixed boundary conditions, which break translational equivariance at the boundaries. In contrast, using periodic boundary conditions on the lattice, translational equivariance can be preserved.

The aforementioned four integer fields of the flux representation are used as the input for the upcoming machine learning tasks. In these tasks, the intensive or extensive nature of an observable are important for the choice of a global pooling layer because the networks should be able to generalize to other lattice sizes apart from the one they have been trained on. In a regression task, an intensive quantity requires a global average pooling layer, while an extensive quantity calls for a global sum pooling layer. In a classification task, it is not the physical observable itself that is predicted but a decision boundary, so the choice of global pooling layer is more subtle.

We are interested in the generalization capability to larger lattices because usually studies on the lattice are intended as an approximation of the real case of an infinite spacetime background. Therefore, a compromise has to be found between lattice size and computational effort such that the simulation produces results satisfactorily close to the physical ones in a reasonable amount of time. The approach adopted throughout this paper is to train models on small lattices and examine how well they generalize to larger lattices. 

The three machine learning tasks that are described at the end of section~\ref{sec:introduction} are tackled with three different types of CNN architecture, which are depicted in Fig.~\ref{fig:architectures}:
\begin{itemize}
    \item a translationally equivariant architecture (\EQ{}) that uses only layers of stride one and a global pooling layer and is applicable to different lattice sizes,
    \item a strided architecture (\ST{}) that breaks translational equivariance due to spatial pooling layers with a stride greater than one but is still suited to give predictions on different lattice sizes due to a global pooling layer, and
    \item a flattening architecture (\FLAT{}) that represents a ``traditional'' architecture that breaks translational equivariance due to spatial pooling layers with a stride greater than one and a flattening layer, which restricts its usage to one particular lattice size.
\end{itemize}
After the global pooling or flattening step, a dense feed-forward network, which is also known as multi-layer perceptron (MLP), can be attached without altering the translational equivariance properties of the network. A discussion about which network layers are translationally equivariant and which ones are not, as well as what exactly breaks said equivariance, can be found in section~\ref{sec:SymmetryProperties}.

Dense neural networks can be seen as universal function approximators~\cite{Cybenko:1989, Lu:2017} and as such, they do not respect any particular symmetries. Such symmetries can be implemented in (or ``hard-baked'' into, as other authors~\cite{Worrall:2017} call it) the network architecture, though, so that the function that is learned is restricted to respect a certain symmetry by design, independent of any training. We expect that such a restriction, as incorporated by \EQ{} architectures, is beneficial and that CNNs of that kind therefore outperform CNNs of the other two kinds.

Alternatively, symmetries can be learned, which can be encouraged by augmenting the training data according to the desired symmetry transformation. Thus, we expect data augmentation to improve the performance of networks of \ST{} and \FLAT{} architectures. Note, however, that when using data augmentation, it is not guaranteed that the network respects the symmetry even on the training set, let alone on the test set. In addition, data augmentation only establishes a relation between the input layer and the output layer; it does not require the hidden layers in between to respect the symmetry.

From a more theoretical standpoint, a CNN can be seen as a special case of an MLP, where the latter learns to set its weights so that the receptive fields are local (by setting the other weights to zero) and to ``share'' the appropriate weights (by setting them to the same value). The idea behind the CNN was, though, that this does not have to be learned but can be implemented in the architecture.

Note that if the observable to be studied violates the symmetry that is implemented in the network, it cannot be properly approximated by design. Therefore, it is important to understand the symmetry properties of the task before selecting a particular architecture.

\section{Symmetry properties of machine learning layers} \label{sec:SymmetryProperties}

If a network layer's output is invariant under a symmetry transformation of the input, the outputs of all subsequent layers are invariant under the symmetry transformation of the former layer's input as well. However, invariance of the network's prediction does not require every individual layer to be invariant under this symmetry. The more general concept of equivariance not only is a sufficient requirement, but also allows for more expressive networks. Group Equivariant Convolutional Neural Networks (G-CNNs)~\cite{Cohen:2016} exploit symmetry transformations that are described by a group~$\mathcal{G}$. Conventional CNNs can be seen as a special case of G-CNNs, with the translation group~$\mathbb{T}$ as their symmetry group, i.e.~\mbox{$\mathcal{G} = \mathbb{T}$}.

The following discussion about equivariance is based on~\cite{Cohen:2016}. The condition for equivariance of a network layer~$\Phi$ under a group transformation~$L_g$ by the element $g\in \mathcal{G}$ is given by
\begin{equation}
        \Phi(L_{\! g} \, x) = L^{\mspace{2mu} \prime}_{\! g} \, \Phi(x). \label{layers:equivariance}
\end{equation}
Note that \mbox{$L_g \neq L^{\mspace{2mu} \prime}_{\! g}$}, in general, and that invariance under~$L_g$ is the special case \mbox{$L^{\mspace{2mu} \prime}_{\! g} = 1$}.

\subsection{Convolutional layers}
	
On a two-dimensional rectangular lattice, a convolution\footnote{The cross-correlation in signal processing is often referred to as convolution in the machine learning community. For this paper we adopt this nomenclature. Additionally, we disregard a possible bias term $b$ to be added to Eq.~\eqref{layers:cross_correlation} without loss of generality.
}
is defined as
\begin{equation}
	[f \star \psi](x) = \sum_{y \in \mathbb{Z}^2} f(y) \psi(y - x) = \sum_{y \in \mathbb{Z}^2} f(x + y) \psi(y), \label{layers:cross_correlation}
\end{equation}
where the feature map $f$ and the kernel (or filter) $\psi$ are real-valued functions
\begin{align}
f : \mathbb{Z}^2 &\rightarrow \mathbb{R}, \\
\psi : \mathbb{Z}^2 &\rightarrow \mathbb{R}.
\end{align}
The kernel $\psi$ is assumed to have finite support $\Psi \subset \mathbb{Z}^2$, i.e.~there is only a finite number of points on $\mathbb{Z}^2$ where $\psi$ is non-zero. In principle, this allows us to restrict the sum on the RHS of Eq.~\eqref{layers:cross_correlation} to $\Psi$, but for simplicity we keep the sum running over all of $\mathbb{Z}^2$. For our purposes, the feature map $f$ is understood to be defined on a finite, rectangular proper subset $F \subset \mathbb{Z}^2$ with periodic boundary conditions. We can avoid explicitly dealing with periodic boundaries by assuming that the feature map periodically repeats outside $F$. In machine learning frameworks such as \textit{PyTorch} \cite{Paszke:2019}, periodic boundary conditions are enforced through the use of circular padding. As a result of periodicity, the output of the convolution Eq.~\eqref{layers:cross_correlation} has the same size as the feature map $f$. 

We define a translation of the feature map via
\begin{align}
[L_t f](x) = f(x - t),
\end{align} 
where $t \in \mathbb{T}$ is an element of the translation group, which can be identified with an element of $\mathbb{Z}^2$. The convolution is equivariant under translations due to
\begin{align}
	[L_t f \star \psi](x) &= \sum_{y \in \mathbb{Z}^2} f(y - t) \psi(y - x) \nonumber \\
	&= \sum_{y' \in \mathbb{Z}^2} f(y') \psi(y' - (x - t)) \nonumber \\
	&=	[f \star \psi](x - t) \nonumber \\
	&=	[L_t [f \star \psi]](x).
\end{align}

Equation \eqref{layers:cross_correlation} assumes the convolution to have a stride~$s$ of one, i.e.~the number of points that the kernel is shifted when the convolution is performed is one. More generally, convolutions with strides \mbox{$s \ge 1$} can be written as
\begin{equation}
	[f \star \psi]_s(x) =  \sum_{y \in \mathbb{Z}^2} f(y) \psi(y - sx). \label{layers:cross_correlation_strides}
\end{equation}
This definition reduces to the original convolution if \mbox{$s=1$}. For \mbox{$s \geq 2$}, the output size of the convolution is smaller than the input size of the feature map~$f$. Strided convolutions with \mbox{$s \geq 2$} generally break translational equivariance. This can be demonstrated by considering a translation \mbox{$t \in \mathbb{T}$} with \mbox{$|t| < s$}. For example, we can choose \mbox{$t= (1, 0)$}. Performing this translation on the input feature map~$f$ yields
\begin{align}
	[L_{t} f \star \psi]_s (x) &= \sum_{y \in \mathbb{Z}^2} f(y - t) \psi(y - s x) \nonumber \\
 &= \sum_{y' \in \mathbb{Z}^2} f(y') \psi(y' - s x + t) \nonumber \\
  &= \sum_{y' \in \mathbb{Z}^2} f(y') \psi(y' - s (x - t / s)).
\end{align}
In order for the above expression to be equivariant, we would need to be able to rewrite it in terms of a shifted position $x' = x - t/s \in \mathbb{Z}^2$. However, this is not possible because $t = (1, 0)$ is not divisible by $s \geq 2$. On the other hand, the strided convolutions are equivariant if we consider only the subgroup $\mathbb T_s \subset \mathbb T$ consisting of translations by multiples of $s$ lattice points. In that case, any element $t \in \mathbb T_s$ is divisible by $s$ and therefore
\begin{align}
	[L_{t} f \star \psi]_s (x) &= \sum_{y \in \mathbb{Z}^2} f(y - t) \psi(y - s x) \nonumber \\
	&= \sum_{y' \in \mathbb{Z}^2} f(y') \psi(y' - s (x - t / s)) \nonumber \\
		&= \sum_{y' \in \mathbb{Z}^2} f(y') \psi(y' - s x') \nonumber \\
				&=[L_{t'} [f \star \psi]_s](x),
\end{align}
where $t' = t / s \in \mathbb T$.  This means that a convolutional layer with a given stride is equivariant only under translations that are a multiple of that stride. Equivariance under all possible translations is given only for~\mbox{$s = 1$}. The generalization to more than one feature map, i.e.~multiple channels, is straightforward. Note that a convolution with~$s \ge 2$ is equivalent to a convolution with~\mbox{$s = 1$} combined with a subsequent subsampling step.

\subsection{Spatial pooling layers}

Spatial pooling layers are usually used to subsample, i.e.~\mbox{$s \ge 2$} in pooling layers. For this discussion, let us split this layer up into a pooling step and a subsampling step. Since average pooling is equivalent to a special case of a convolution, where all weights of~$\psi$ are identical and given by~\mbox{$1 / | \Psi |$}, with $|\Psi|$ denoting the cardinality of~\mbox{$\Psi$}, the average pooling step is equivariant under translations. The subsequent subsampling, however, breaks this equivariance, which again leads to equivariance only under translations that are a multiple of the spatial average pooling layer's stride.

This holds not only for average pooling though, but for spatial pooling in general: we take again the pooling step by itself, or equivalently, with~\mbox{$s = 1$}. It acts on the feature map~$f$ by performing the same operation on subsets~$U_x$ of $F$
\begin{equation}
    P f(x) = \mspace{-3mu} \underset{y \in U_x}{P} \mspace{-3mu} f(y).
\end{equation}
These subsets correspond to the kernel of the pooling operation. Its dependence on~$x$ depicts the ``sliding'' of the kernel over the feature map. A spatial pooling step respects Eq.~\eqref{layers:equivariance}, as can be seen by
\begin{align}
    P L_t f(x) &= \mspace{-3mu} \underset{y \in U_x}{P} \mspace{-3mu} f(y - t) \nonumber \\
    &= \mspace{-10mu} \underset{y' \in U_{x - t}}{P} \mspace{-10mu} f(y') \nonumber \\
    &= L_t P f(x).
\end{align}
Thus, also in a spatial pooling layer it is the stride that restricts the equivariance of the layer to translations by multiples of said stride.

We want to stress that spatial pooling layers with~\mbox{$s = 1$} respect translational equivariance and can therefore be included if one desires an architecture that incorporates such a symmetry, albeit in a different role than usual because it does not subsample.

\subsection{Global pooling}

If we wanted to use a traditional CNN architecture on a two-dimensional lattice with periodic boundary conditions, we would have another problem as well: the last convolutional or pooling layer is often flattened and densely connected to the linear layers at the end of the network. Since different positions in one feature map are connected to different weights without a ``sliding'' kernel, this is another point where translational equivariance is broken. A possible solution to this problem is a global pooling layer between the last convolution and the first dense layer. The GAP layer was first introduced in~\cite{Lin:2014}. There, the authors proposed to create one feature map for each class and to feed the average of each feature map directly to a softmax layer. This approach would respect translational symmetry, although, in general, dense layers could be used between the global pooling and the softmax operation.

\subsection{Equivariant architectures}

On the aforementioned two-dimensional lattice with periodic boundary conditions and for similar problems we propose the following network architecture for classification and regression tasks: the input is fed to a convolutional layer with a stride of one and circular padding so that the output of the convolution has the same size as its input. The kernel size can be odd or even. Translational equivariance is retained by applying consecutive convolutional layers, all with~\mbox{$s = 1$}, with non-linear activation functions in between. Activation functions do not influence the symmetries of an individual layer, since they are applied pointwise. If information from different scales is required, dilated convolutions~\cite{Yu:2016} can be used with a stride of one. Since dilated convolutions are equivalent to convolutions with a larger kernel and the appropriate weights set to zero, they are also equivariant under translations if their stride is one. Spatial pooling layers for subsampling, which use $s > 1$, break translational equivariance, but it is still possible to use them with \mbox{$s = 1$} between convolutional layers. A way of subsampling that respects translational equivariance is rendered possible by coset pooling~\cite{Cohen:2016}. However, since this is a non-local operation, we do not expect it to be suitable for the machine learning tasks discussed in this paper, which focus on local quantities and predictions. In the special case of translationally invariant functions, we suggest to utilize a global pooling layer after the last convolution. The output of the global pooling layer is translationally invariant, and therefore the rest of the network can be a general MLP without breaking the symmetry.

There is still one important point to be made: every layer before the GAP respecting translational equivariance is sufficient to guarantee invariance under translations after the GAP, but it is not necessary. If a spatial average pooling layer that breaks translational equivariance and a subsequent convolutional layer are inserted just before the GAP, the output of the GAP can still be invariant under translations, depending on their strides (theorem~\ref{proofs:theorem:SpatConvGAP} in appendix~\ref{app:Proofs}). If there is an activation function after the convolutional layer, as is usually the case, the GAP's output is, in general, no longer invariant under translations. The activation function is also necessary for the convolution not to lead to a single multiplicative and additive factor of the GAP, as is shown in lemma~\ref{proofs:lemma:ConvAndGAP} in appendix~\ref{app:Proofs}. We thus stick to the aforementioned sufficient conditions for translational equivariance and apply an activation function after the convolutional layer right before the GAP.

\section{Regression: predicting observables on the lattice} \label{sec:Regression}

This section revisits a regression task that has previously been performed in~\cite{Zhou:2019}: given a lattice configuration as input, the network shall predict two physical observables, namely the particle density~$n$ and the lattice averaged squared absolute value of the field~$|\phi|^2$. The former is given by
\begin{equation}
    n=\frac{1}{N_x N_t}\sum_x k_{x, 2}, \label{Regression:n_formula}
\end{equation}
where the summation of one of the four integer fields~$k_{x,2}$ runs over all $N_x N_t$ lattice sites. The latter is given by
\begin{equation}
    |\phi|^2 = \frac{1}{N_x N_t}\sum_x\frac{W(f_x+2)}{W(f_x)}, \label{Regression:phi2_formula}
\end{equation}
which contains the highly non-linear function~$W(f_x)$. It is given in Eq.~\eqref{trans:W} and depends on all four integer fields.

The function~$W(f_x)$ also depends on the physical parameters~$\lambda$, $\eta$ and~$\mu$, which are set to the same values as in~\cite{Zhou:2019}. Concretely, the values of the coupling constant~$\lambda$ and the mass~$m$ will be kept fixed in this task ($\lambda=1$, $\eta=4+m^2=4.01$), and the chemical potential~$\mu$ lies in the interval~\mbox{$\mu \in [0.91, 1.05]$}, with steps of \mbox{$\Delta \mu = 0.005$}.

In~\cite{Zhou:2019}, the networks have been trained on lattice configurations and observables that have been generated with two values of~$\mu$, specifically the outermost values~\mbox{$\mu = \{ 0.91, 1.05 \}$}, but tested on data that have been created on the whole given interval of the chemical potential. This allows for an analysis of the architectures' generalization capability to lattice configurations that correspond to chemical potentials that are not represented in the training set. We will follow this procedure, with the exception that we will use only a single~$\mu$ to generate training data, namely the uppermost one~\mbox{$\mu = 1.05$}. Since we would test on only smaller values of the chemical potential than the one that has been used for training in our approach, we deviate from~\cite{Zhou:2019} in that we create additional test data that contain higher values of the chemical potential. They lie in the interval~\mbox{$\mu \in [1.1, 1.5]$}, with steps of \mbox{$\Delta \mu = 0.1$}. This renders possible an analysis of the architectures' generalization capability to lattice configurations that correspond to values of~$\mu$ that are greater than the one used to create the training set. We will come back to these test data only at the end of this section. In addition to the generalization ability to different values of the chemical potential, we will investigate the generalization ability to lattice sizes that the models have not been trained on. This highlights a key advantage of architectures that employ a global pooling layer between their convolutional and their dense layers over architectures that simply flatten the data, because the latter are restricted to a given input size.

\subsection{Architecture choice}

The datasets stem from a physical system, whose properties should be taken into account when choosing a network architecture for a model that should learn from said dataset. Let us assume for the following discussion that we have no knowledge of the exact form of Eqs.~\eqref{Regression:n_formula} and~\eqref{Regression:phi2_formula}.

First, the observables are invariant under arbitrary translations of the lattice configuration. This leads to the restriction of preferred architectures that has been proposed at the end of section~\ref{sec:SymmetryProperties}: the input is passed to a convolutional layer with a stride~$s = 1$ and circular padding that causes its output to have the same size as its input. Such layers are used consecutively, with non-linear activation functions in between. Optionally, spatial pooling layers with \mbox{$s = 1$} can be inserted. The output of such convolutional and pooling layers is equivariant under translations of the input. The output of the last convolution is fed to a global pooling layer, which makes it invariant under translations of the input. Then, the data are passed through an MLP with two output nodes, one for each observable.

Second, the observables are derivatives of the logarithm of the partition function on the lattice. The partition function is a product over quantities at each lattice site. The observables can therefore be written as a sum over the lattice. Consequently, we want to use a global pooling layer that respects this fact, which excludes global max pooling. Since the observables are intensive quantities and the network shall be able to generalize to different lattice sizes, global average pooling is the natural choice. The MLP at the end does not modify the intensive nature of the prediction. 

To check how the above theoretical considerations perform in practice, we want to compare the three types of architecture that are depicted in  Fig.~\ref{fig:architectures} from section~\ref{sec:translational_symmetry}. A fair comparison among these network architecture types is quite difficult. One could take a translationally equivariant architecture and break equivariance by inserting at least one spatial pooling layer with~\mbox{$s > 1$}. This would keep the number of parameters the same. However, having found a decent \EQ{} architecture, it is not guaranteed that the corresponding \ST{} architecture is a good one compared to other \ST{} architectures and vice versa. Also, keeping the weights constant may not lead to a fair comparison with \FLAT{} architectures.

Therefore, we define a space of possible architectures for each of the three types separately, which are illustrated in tables~\ref{tab:Regression:optuna_search_space_equiv} to~\ref{tab:Regression:optuna_search_space_non-equiv_FL}, and use an optimization procedure to find an adequate representative for each architecture type individually.

Table~\ref{tab:Regression:optuna_search_space_equiv} depicts the search spaces of \EQ{}, table~\ref{tab:Regression:optuna_search_space_non-equiv_GP} of \ST{} and table~\ref{tab:Regression:optuna_search_space_non-equiv_FL} of \FLAT{} architectures. The possible parameter values of the first run are inspired by manual trials, which also included different activation functions (\textit{ReLU, tanh, PReLU} and \textit{LeakyReLU}). Its results lead to modifications of the parameter space of the second run and the choice of \textit{LeakyReLU} for a suitable activation function. Both of them try $50$ different combinations of parameters in their optimization procedure on each training set, which will be specified in the next subsection. The extended search explores an enlarged parameter space with $100$~trials, also with unique combinations of parameter values, in order to check if a better architecture was missed during the first two runs due to the choice of a too small search space. This search involves only the largest training set. After every convolutional layer and after every linear layer but the last one, a \textit{LeakyReLU} activation function~\cite{Maas:2013} is applied. Its advantage over the \textit{ReLU} activation function is the avoidance of so-called dead or dying neurons, which never activate initially or become inactive during the training process.

\begin{table}[tbp]
    \centering
    \scriptsize
    \caption{Search spaces for \EQ{} architectures. It lists the possible number of convolutional (conv, \mbox{$s=1$}) and linear layers (lin), kernel sizes, the number of channels of the convolutional layers and the number of nodes in the linear layers. Spatial pooling layers with \mbox{$s=1$} seem to worsen the predictions and have therefore not been included in these search spaces.}
    \begin{ruledtabular}
    \begin{tabular}{lllll}
        & conv & lin & kernel size & channels/nodes \\
        \hline
        run 1 & $[2,3]$ & $[0,1]$ & $\{ (1 \mspace{-5mu} \times \mspace{-5mu} 1), (2 \mspace{-5mu} \times \mspace{-5mu} 2) \}$ & $\{4,8,16,24,32,48,64,80\}$ \\
        run 2 & $[2,4]$ & $1$ & $\{ (1 \mspace{-5mu} \times \mspace{-5mu} 1), (2 \mspace{-5mu} \times \mspace{-5mu} 2) \}$ & $\{4,8,16,24,32,48,64,80\}$ \\
        extended search & $[2,4]$ & $[0,3]$ & $\{ (1 \mspace{-5mu} \times \mspace{-5mu} 1), (2 \mspace{-5mu} \times \mspace{-5mu} 2) \}$ & $\{4,8,16,24,32,48,64,80\}$ \\
    \end{tabular}
    \end{ruledtabular}
    \label{tab:Regression:optuna_search_space_equiv}
\end{table}

\begin{table*}
    \centering
    \caption{Search spaces for \ST{} architectures. It shows the possible number of convolutional \mbox{$(s=1)$} and linear layers (abbreviated as in table~\ref{tab:Regression:optuna_search_space_equiv}), kernel sizes, the number of channels of the convolutional layers, the number of nodes in the linear layers, the number of spatial pooling layers (SPL, \mbox{$s=2$}) and the spatial pooling mode (SPM).
    }
    \begin{ruledtabular}
    \begin{tabular}{lllllll}
        & conv & lin & kernel size & channels/nodes & SPL & SPM \\
        \hline
        run 1 & $[2,4]$ & $[0,3]$ & $\{(1 \times 1), (2 \times 2)\}$ & $\{4,8,16,24,32,48,64,80\}$ & $\{1,2\}$ & $\{\mathrm{avg}, \mathrm{max}\}$ \\
        run 2 & $[2,4]$ & $[0,2]$ & $\{(1 \times 1), (2 \times 2)\}$ & $\{4,8,16,24,32,48,64,80\}$ & $\{1,2\}$ & avg \\
        extended search & $[2,4]$ & $[0,3]$ & $\{(1 \times 1), (2 \times 2)\}$ & $\{4,8,16,24,32,48,64,80\}$ & $\{1,2\}$ & $\{\mathrm{avg}, \mathrm{max}\}$ \\
    \end{tabular}
    \end{ruledtabular}
    \label{tab:Regression:optuna_search_space_non-equiv_GP}
\end{table*}

\begin{table*}
    \centering
    \caption{Search spaces for \FLAT{} architectures. It shows the possible number of convolutional \mbox{$(s=1)$} and linear layers, kernel sizes, the number of channels of the convolutional layers, the number of nodes in the linear layers, the number of spatial pooling layers \mbox{$(s=2)$} and the spatial pooling mode. The number of convolutional layers is not chosen directly but follows from the number of $1\times 1$ convolutions that are selected. The asterisk next to ``kernel size'' signifies that the kernel size of the convolution depends on its position. Two $2 \times 2$ convolutions with a respective subsequent spatial pooling layer are mandatory. Additional $1 \times 1$ convolutions are possible, namely before each of the $2 \times 2$ convolutions and between each of them and their corresponding subsequent spatial pooling layer. The abbreviations are the same as in table~\ref{tab:Regression:optuna_search_space_non-equiv_GP}.}
    \begin{ruledtabular}
    \begin{tabular}{lllllll}
        & conv & lin & kernel size$^*$ & channels/nodes & SPL & SPM \\
        \hline
        run 1 & $[2,6]$ & $[1,3]$ & $\{(1 \times 1), (2 \times 2)\}$ & $\{4,8,16,24,32,48,64,80\}$ & $2$ & $\{\mathrm{avg}, \mathrm{max}\}$ \\
        run 2 & $[2,6]$ & $[1,3]$ & $\{(1 \times 1), (2 \times 2)\}$ & $\{4,8,16,24,32,48,64,80\}$ & $2$ & avg \\
        extended search & $[2,6]$ & $[1,3]$ & $\{(1 \times 1), (2 \times 2)\}$ & $\{4,8,16,24,32,48,64,80\}$ & $2$ & $\{\mathrm{avg}, \mathrm{max}\}$ \\
    \end{tabular}
    \label{tab:Regression:optuna_search_space_non-equiv_FL}
    \end{ruledtabular}
\end{table*}

\ST{} architectures can be thought of as \EQ{} architectures with at least one spatial pooling layer with~\mbox{$s = 2$} in the convolutional part of the network. This is either an average pooling or a max pooling layer, both with a $2 \times 2$~kernel. A spatial pooling layer is neither directly applied to the input, nor inserted just before the global pooling layer. The position of the spatial pooling layer(s) is part of the search space, but it is restricted by the choice of the number of convolutional layers, as is the number of spatial pooling layers. If, e.g., two convolutional layers are chosen, there can only be one spatial pooling layer at only one specific position, that is between the convolutional layers.

\FLAT{} architectures are inspired by how we think one would construct a CNN traditionally for this machine learning problem. At its core are two \mbox{$2 \times 2$} convolutions, followed by a spatial pooling layer with a \mbox{$2 \times 2$} kernel and a stride of~$2$. Optionally, there can be a $1 \times 1$ convolution before each of the \mbox{$2 \times 2$} convolutions and between each of them and the respective following spatial pooling layer, leading to a possible total count of six convolutional layers.

Our optimization procedure of choice has been \textit{optuna}. The performance metric is the validation loss averaged over three different parameter initializations. This averaging process is applied to counteract the statistical fluctuations introduced by the random initializations of the trainable network parameters. It is important because \textit{optuna} changes its search space dynamically, so early search results influence the probability distributions that serve as the basis to select later parameter values. This optimization process is done for different sized training sets individually, since on smaller training sets different architectures might perform better than on larger ones.

After the optimization procedure by \textit{optuna}, models of the best architectures are retrained ten times from scratch and evaluated on the validation set to verify their performance while further minimizing statistical fluctuations due to the random parameter initializations. Our results show that the same architectures that perform well on small training sets also perform well on larger training sets and that many architectures perform similarly. Thus, we select the best-performing architecture of each type according to the mean validation loss as a representative and compare only them. These best-performing architectures are shown in table~\ref{tab:best_performing_architectures}. We use \mbox{Conv($K \times K$, $N_\mathrm{in}$, $N_\mathrm{out}$)} to denote a two-dimensional convolution where $K$ is the kernel size and $N_\mathrm{in}$ ($N_\mathrm{out}$) is the number of input (output) channels. Before every convolutional operation we use circular padding to enforce periodic boundary conditions. Additionally, we use a stride of one for each convolution. Average pooling layers with kernel size $K$ and stride $s$ are written as \mbox{AvgPool($K \times K$, $s$)}. Dense layers are denoted by \mbox{Linear($N_\mathrm{in}$, $N_\mathrm{out}$)} with $N_\mathrm{in}$ ($N_\mathrm{out}$) input (output) nodes. 

\begin{table}[htbp]
\centering
\scriptsize
\caption{\label{tab:best_performing_architectures} 
Best architectures for fitting two observables $n$ and \mbox{$\lvert \phi \rvert^2$} for each type of architecture. This table shows feed-forward networks as found by our \textit{optuna} searches. The field configuration in the form of $(N_t, N_x, 4)$ tensors is fed into the network at the top. (The batch size is omitted here.) There are two output nodes for the two observables. The last row denotes the number of trainable parameters for each type.}
\begin{ruledtabular}
\begin{tabular}{lll}
\textbf{\EQ{}} & \textbf{\ST{}}  & \textbf{\FLAT{}} \\
 \hline
Conv($1 \times 1$, 4, 64)  & Conv($1 \times 1$,  4, 80)         & Conv($1 \times 1$, 4, 64)           \\
LeakyReLU                  & LeakyReLU                          & LeakyReLU                          \\
Conv($1 \times 1$, 64, 48) & Conv($1 \times 1$, 80, 80)         & Conv($2 \times 2$, 64, 80)         \\
LeakyReLU                  & LeakyReLU                          & LeakyReLU                          \\
Conv($1 \times 1$, 48, 80) & Conv($1 \times 1$, 80, 48)         & AvgPool($2 \times 2$, $2$)         \\
LeakyReLU                  & LeakyReLU                          & Conv($1 \times 1$, 80, 48)         \\
Conv($2 \times 2$, 80, 80) & AvgPool($2 \times 2$, $2$)         & LeakyReLU                          \\
LeakyReLU                  & Conv($2 \times 2$, 48, 80)         & Conv($2 \times 2$, 48, 64)         \\
GlobalAvgPool              & LeakyReLU                          & LeakyReLU                          \\
Linear(80, 2)              & GlobalAvgPool                      & AvgPool($2 \times 2$, $2$)         \\
                           & Linear(80, 2)                      & Conv($1 \times 1$, 64, 24)         \\
                           &                                    & Flatten                            \\
                           &                                    & Linear(360, 24)                    \\
                           &                                    & LeakyReLU                          \\
                           &                                    & Linear(24, 2)                      \\
\hline
33202                      & 26370                              & 47394                              \\
\end{tabular}
\end{ruledtabular}
\end{table}

\subsection{Training and testing}

The training is performed for every model of each of these three architectures and for each training set analogously: mean squared error (MSE) is selected as a loss function; the total loss is the arithmetic mean of the individual losses, each of which corresponds to one physical observable. It is optimized with the \textit{AMSGrad}~\cite{Reddi:2019} variant of the \textit{AdamW} optimizer~\cite{Loshchilov:2019} with a vanishing weight decay. Training models on different sized training sets gives us information about the sample efficiency. Limiting the size of training sets is motivated by machine learning tasks where the generation of training samples is costly, for example in medical applications or in large-scale simulations on supercomputers. The number of training samples in a training set ranges from~$100$ to~$20000$, with steps~$\Delta = 50$ from~$100$ to~$250$, $\Delta = 250$ from~$250$ to~$1000$, $\Delta = 500$ from~$1000$ to~$3000$ and $\Delta = 1000$ up to~$20000$ training samples. The corresponding validation sets contain~$10\%$ of the amount of the training set's data. The batch size during training was chosen to be~$100$ for training sets with at least~$500$ training samples and~$50$ otherwise. The reason behind this choice is that the algorithm shall be trained with mini-batches. To avoid that this approaches batch training for smaller training sets, a smaller batch size is chosen for them. The training lasts between~$100$ and~$1000$ epochs; the exact number is determined by early stopping based on validation loss with a \textit{patience} value of $25$. The model is taken at the time it has had the lowest validation loss. An overview of the chosen parameters is given in table~\ref{tab:Regression:parameters_for_training}.

\begin{table}[tbp]
    \centering
    \scriptsize
    \caption{Loss, optimizer and early stopping settings for \mbox{\textit{PyTorch}}.}
    \begin{ruledtabular}
    \begin{tabular}{llllll}
        loss & size\_avg & reduce & \multicolumn{3}{l}{reduction} \\
        MSELoss & None & None & \multicolumn{3}{l}{`mean'} \\
        \hline
        optimizer & lr & betas & eps & weight\_decay & amsgrad  \\
        AdamW & $0.001$ & $(0.9, 0.999)$ & $10^{-8}$ & 0 & True \\
        \hline
        \multicolumn{1}{l}{} & monitor & min\_delta & patience & \multicolumn{2}{l}{mode} \\
        EarlyStopping & `val\_loss' & $0$ & $25$ & \multicolumn{2}{l}{`min'}  \\
    \end{tabular}
    \end{ruledtabular}
    \label{tab:Regression:parameters_for_training}
\end{table}

The training takes place on a $60 \times 4$ lattice; the first number refers to the temporal dimension and the second to the spatial one. All data in the training set and the validation set have been generated with~$\mu = 1.05$.

Both translationally non-equivariant architectures (\ST{} and \FLAT{}) are trained with and without data augmentation. The training data are augmented by randomly shifting the input data by a number of pixels that is determined by the symmetry properties of the respective architecture. \ST{} architectures contain at most two spatial pooling layers with a \mbox{$2 \times 2$} kernel and a stride of~$2$, as is shown in table~\ref{tab:Regression:optuna_search_space_non-equiv_GP}. Therefore, they still incorporate translational equivariance under shifts of multiples of~$4$ (see section~\ref{sec:SymmetryProperties}); and the data can be augmented by shifts of~$[0,3]$ in both directions. \FLAT{} architectures, however, do not incorporate translational equivariance under any shifts of the input; thus, the data have to be augmented by shifts determined by the lattice size, i.e.~by~$[0,59]$ in the time direction and by~$[0,3]$ in the space direction.

The testing can be divided into two parts. As a first step, each architecture is evaluated on the same lattice size as it has been trained on, for various values of~$\mu$. This checks whether networks of a given architecture are able to generalize to values of~$\mu$ that are not represented in the training set. Then, the generalization ability to other lattice sizes is investigated. This second step can be done only with architecture types \EQ{} and \ST{} because \FLAT{} architectures require a fixed input size.

The test set on the~\mbox{$60 \times 4$} lattice contains samples that have been generated with various values of~$\mu$, most of which have not been used for the training and validation sets. This test set contains $4000$~lattice configurations pertaining to each~\mbox{$\mu \in [0.91, 1.05]$}, with steps of \mbox{$\Delta \mu = 0.005$}, where only the last value~$\mu = 1.05$ has been used for training and validation. This amounts to~\mbox{$1.16 \times 10^5$} testing samples in total.

For testing on different lattice sizes, we generated a test set analogous to the one on the~\mbox{$60 \times 4$} lattice on a~\mbox{$50 \times 2$}, a~\mbox{$100 \times 5$}, a~\mbox{$125 \times 8$} and a~\mbox{$200 \times 10$} lattice. For each of these lattice sizes, we created again~\mbox{$1.16 \times 10^5$} test samples, $4000$ pertaining to each~\mbox{$\mu \in [0.91, 1.05]$}, with steps of \mbox{$\Delta \mu = 0.005$}. Note that the winning \ST{} architecture (see table~\ref{tab:best_performing_architectures}) can be evaluated on the~$50 \times 2$ lattice because it contains only one spatial pooling layer with a \mbox{$2 \times 2$}~kernel and a stride~\mbox{$s = 2$}.
Further details on the dataset generation can be found in appendix~\ref{app:Datasets}.

\subsection{Results}

In this subsection, we will discuss the test results on the \mbox{$60 \times 4$}~lattice, which is the lattice size on which the training took place, in detail before analyzing the generalization ability to other lattice sizes of the different network types. Then, we will investigate the Silver Blaze phenomenon on the larger lattice sizes with our trained models. Finally, we will discuss the results on our second set of test sets, which contains data generated with a chemical potential greater than the one used to create the training set.

\subsubsection{Results on the same lattice size as training}

\begin{figure}[tbp]
	\centering
	\includegraphics{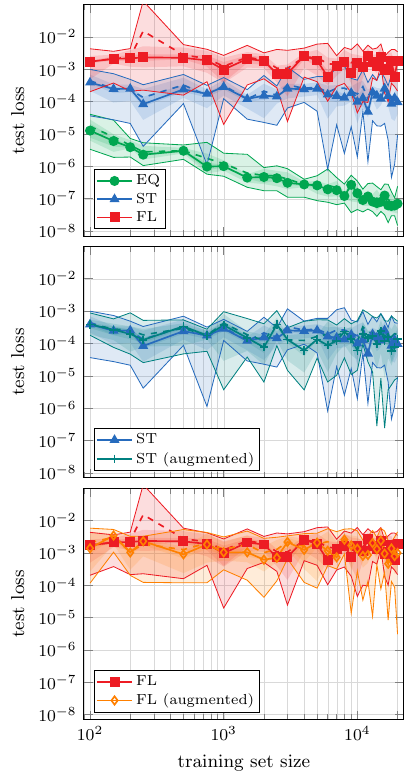}
	\caption{Test loss on the whole test set on the \mbox{$60 \times 4$}~lattice against the size of the training set (number of samples in the training set) on which the respective model has been trained. At the top, the results of the three architecture types (trained without data augmentation) are shown. In the middle and at the bottom, the effect of data augmentation during training of \ST{} and \FLAT{} models, respectively, is depicted. The plots display the best and worst loss (solid lines), the arithmetic mean of all ten random initializations for training (dashed lines) and the $20\%$ quantiles (shaded regions). The symbols visualize the positions of the measurements; the lines are there to guide the eye.}
	\label{fig:test_loss_60_times_4}
\end{figure}

The loss over the whole test set is a metric for how well the network performs. It is displayed in Fig.~\ref{fig:test_loss_60_times_4} for different training sets with a varying number of training samples. Essentially all of our models are trained until convergence, since we choose a very high number of maximum epochs and employ early stopping based on validation loss. Therefore, the comparison in Fig.~\ref{fig:test_loss_60_times_4} shows how the different architectures perform under limited information for smaller training set sizes. The plot at the top shows that the performance of the \EQ{} architecture improves with the size of the training set, as can be expected. The other two architectures do not seem to benefit from increasing the number of samples in the training set, which is quite surprising. Another remarkable result is that data augmentation does not seem to lead to an increase in performance either, as can be seen in the plot in the middle and at the bottom. At first sight, one may draw the conclusion that the \ST{} and the \FLAT{} architectures do not allow for approximations that are as precise as the one of the \EQ{} architecture. If the model has already converged to an optimal solution, adding more training samples,  irrespective of them being newly created or coming from data augmentation, will not improve its performance. However, the blue downward spikes in the loss of the \ST{} model show that some models succeed in finding a good approximation of the observables. Therefore, we draw the conclusion that although possible, it is more unlikely for the \ST{} and \FLAT{} models than for the \EQ{} models to learn a good approximation of~$n$ and~$\lvert \phi \rvert^2$.

\begin{figure}[htbp]
	\centering
	\includegraphics{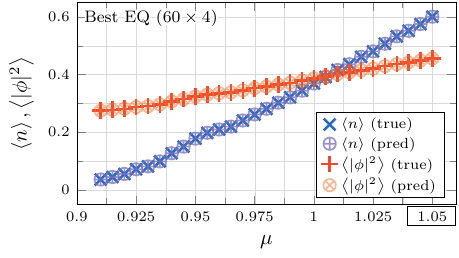}
	\caption{Predicted and true values for ensemble averages $\left< n \right>$ and $\left< \lvert \phi \rvert^2 \right>$ as a function of chemical potential $\mu$ on a $60 \times 4$ lattice. The predictions in this plot are made by the \EQ{} model with the smallest test loss. The model has been trained on data generated with $\mu = 1.05$ only but shows remarkable generalization capabilities to other values of $\mu$. In this and in subsequent plots, the training point is highlighted by a rectangle.}
	\label{fig:reg_observables_over_mu}
\end{figure}

The predictions of the individual observable's ensemble averages per $\mu$ made by the best \EQ{} model, according to the test loss, are displayed in Fig.~\ref{fig:reg_observables_over_mu}. It shows that the model, although trained only on samples generated with $\mu = 1.05$, can generalize to all other values of the chemical potential in the investigated interval. This seemingly astonishing generalization ability can be understood by recognizing that the network does not need to generalize from one $\mu$ to all others but from the training samples to other samples, each consisting of a lattice configuration and two observables. Even though the training set contains only lattice configurations that have been generated with~\mbox{$\mu = 1.05$}, the range of possible values for $n$ and $|\phi|^2$ is quite large and the chosen value of~$\mu$ in the training set already covers most of the observable values in the test set.

\begin{figure}[htbp]
	\centering
	\includegraphics{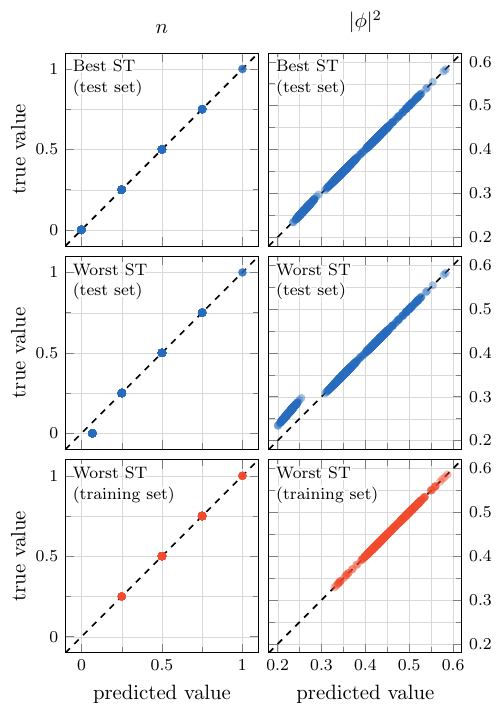}
	\caption{Predicted versus true observables for the best and the worst \ST{} network that have been trained on $18000$ samples. It shows that the \ST{} architecture's best instance is able to accurately estimate the whole ranges of observable values (top) and that its worst instance is failing to do so for smaller values of~$n$ and~$\lvert \phi \rvert^2$ (middle). The reason for this is that the training set includes only larger values of the observables (bottom) and that the worst model is not able to generalize beyond that. The top and the middle plot show $1\%$ of the test data; the bottom plot shows $4\%$ of the training data.}
	\label{fig:test_plots_mu_60_times_4}
\end{figure}

We exemplify this point using \ST{} models that have been trained on $18000$ samples: Figure~\ref{fig:test_plots_mu_60_times_4} shows the predicted versus the true values of both observables of the best (top) and the worst (middle) performing \ST{} model (according to the test loss) evaluated on the test set. The performance of the worst \ST{} model on the training data is shown at the bottom in Fig.~\ref{fig:test_plots_mu_60_times_4}. Note that in this scatter plot we do not average over the ensemble but show the predictions of the network for each individual example. Both networks are able to predict the larger values of both observables, but the worst one fails to predict the smallest values, since they are missing from the training set. The difference between the better and the worse \ST{} models is the ability to generalize to lattice configurations and ranges of observable values that is has not seen during training. The bad performance overall with some better performing outliers, which is shown in Fig.~\ref{fig:test_loss_60_times_4}, suggests that \ST{} networks succeed only sometimes with this generalization. \FLAT{} models show a similar behavior to \ST{} models, but the predictions are less precise throughout. A more detailed discussion of the input value distributions is given in appendix~\ref{app:Datasets}.

\subsubsection{Results on different lattice sizes} \label{sec:RegressionDifferentLatticeSizes}

One big disadvantage of \FLAT{} architectures impedes them from predicting on other lattice sizes than the one it was trained on: it requires a fixed input size. For this reason, we can compare only the performance of the \EQ{} and the \ST{} architecture. Since the results for the latter with and without data augmentation are very similar, we will show only the results without data augmentation. Also, we will fix the size of the training set for this comparison to $20000$~training samples.

\begin{figure}[htbp]
	\centering
	\includegraphics{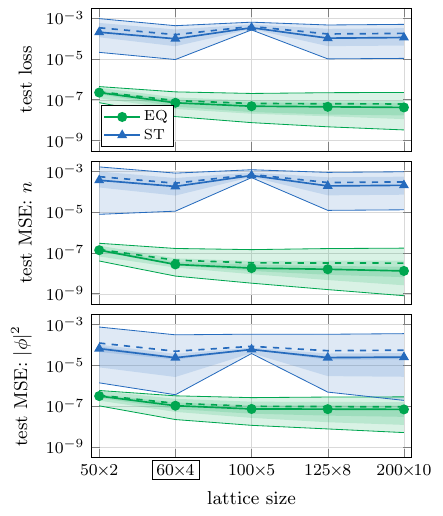}
	\caption{Overall test loss (top) and its two parts (middle and bottom) that come from each observable, on various lattice sizes. The training has taken place on the \mbox{$60\times 4$} lattice. Both architectures generalize well to lattice sizes different from the one they were trained on, but the \ST{} architecture (blue) performs visibly worse on the \mbox{$100 \times 5$}~lattice. The reason for this is the spatial pooling layer within the architecture, which drops $20 \%$ of the data, leading to a less accurate prediction for both observables.}
	\label{fig:test_loss_various_lattice_sizes}
\end{figure}

Figure~\ref{fig:test_loss_various_lattice_sizes} displays the overall test loss (top) and the individual losses of the observables (middle and bottom). Even though the \ST{} architecture keeps its worse performance from the \mbox{$60 \times 4$} lattice, the generalization ability to the different lattice sizes is comparable for the \EQ{} and the \ST{} architecture, with the exception of the \mbox{$100 \times 5$}~lattice for the latter. This kink in the blue curve shows up in the prediction of both observables, whereas this particular lattice size does not seem to be extraordinary to the \EQ{} architecture. The problem is the odd number in the lattice dimension. This behavior can be explained by a closer inspection of the \ST{} architecture: the first three convolutions leave the input size unchanged because they employ circular padding. Then, these \mbox{$100 \times 5$}~data are passed to a spatial pooling layer with a $2 \times 2$~kernel and a stride of~$2$. This layer disregards~$20 \%$ of the data and outputs data with a shape of~\mbox{$50 \times 2$}. Consequently, the \ST{} networks cannot use all of the data to come to a prediction, which is therefore less precise. This is far less severe on the \mbox{$125 \times 8$}~lattice, because there the spatial pooling layer disregards only $1 / 125$ of the data, which is not enough to be visible in Fig.~\ref{fig:test_loss_various_lattice_sizes}. A more detailed analysis of said kink in the blue curve can be found in appendix~\ref{app:PartiallyObscuredInput}.

\subsubsection{Silver Blaze phase transition}

The Silver Blaze~\cite{Cohen:2003} phenomenon refers to a second-order phase transition at vanishing temperature~$T$, where thermodynamical observables are independent of the chemical potential~$\mu$ below a critical value~$\mu_c$~\cite{Gattringer:2013b}. This means that the observables $\langle n \rangle$ and $\langle \lvert \phi \rvert^2 \rangle$ are constant for~\mbox{$\mu < \mu_c$}, whereas they start rising if the chemical potential surpasses its critical value. The particle density~$\langle n \rangle$ is an order parameter of the Silver Blaze phase transition. As a result of the finiteness of our lattices, the temperature is non-zero (\mbox{$T \propto 1 / N_t$}, where~$N_t$ is the number of lattice sites in the time direction), and thus the transition is not necessarily sharp. Because of the networks being trained to approximate the particle density and the lattice averaged squared absolute value of the field, this phase transition should also be visible in their predictions.

\begin{figure}[htbp]
	\centering
	\includegraphics{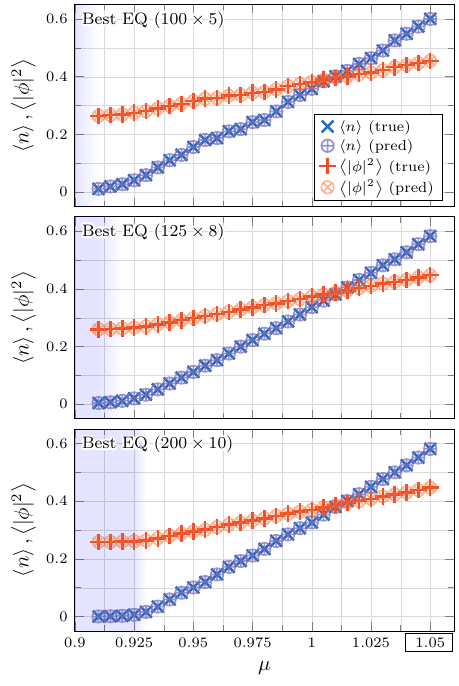}
	\caption{Predicted and true mean values of each observable for each individual~$\mu$ on the larger lattices. The predictions come from the \EQ{} model that has the lowest validation loss from all \EQ{} models that have been trained on $20000$ training samples. The training has been performed at $\mu = 1.05$. The kinks in the curves allow for an estimate of the Silver Blaze phase transition, which is indicated by the color gradient from the shaded region to the white background.} 
	\label{fig:silver_blaze_prediction_larger_lattice_sizes}
\end{figure}

Figure~\ref{fig:silver_blaze_prediction_larger_lattice_sizes} visualizes predictions of the \EQ{} architecture model that has been trained on $20000$~training samples and reached the lowest validation loss. More precisely, it shows the mean prediction of each observable for each individual~$\mu$, as well as the true mean value, on the~\mbox{$100 \times 5$} (top), the~\mbox{$125 \times 8$} (middle) and the \mbox{$200 \times 10$} (bottom) lattice. The largest lattices show both phases, whereas the smaller lattices show no phase transition in the range of~$\mu$ that we analyzed. This is because~$\mu_c$ decreases for increasing temperature.

The Silver Blaze phase transition is also predicted correctly by the \ST{} models that accurately generalize to the smaller values of the observables, e.g.~by the model that is shown at the top in Fig.~\ref{fig:test_plots_mu_60_times_4}, but not all \ST{} models generalize well.

\subsubsection{Extrapolation to larger chemical potentials}

After inspecting the remarkable results that the \EQ{} architecture and some models of the \ST{} architecture achieved on the interval~\mbox{$\mu \in [0.91, 1.05]$}, the question remains as to how the different architectures perform on data corresponding to chemical potentials greater than the one of the training set.\footnote{This was done thanks to a suggestion by the referee.} To answer it, we evaluate the already trained networks on additional test sets, without retraining them. We have created one test set for each lattice size under consideration. Each of them contains $4000$~lattice configurations corresponding to each~\mbox{$\mu \in [1.1, 1.5]$}, with steps of \mbox{$\Delta \mu = 0.1$}. This amounts to~\mbox{$2 \times 10^4$} test samples per lattice size.

\begin{figure}[tbp]
	\centering
	\includegraphics{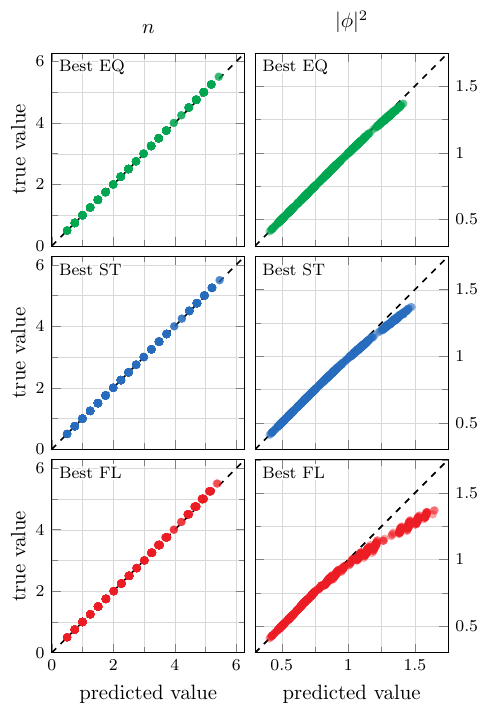}
	\caption{Predicted versus true observables for the best (according to the validation loss) model of each architecture evaluated on the test set generated from \mbox{$\mu \in [1.1, 1.5]$} on the \mbox{$60 \times 4$} lattice. Each model is able to predict higher values of~$n$ than given during training, but the generalization of~$\lvert \phi \rvert^2$ exhibits a clear difference between the generalization capabilities of the models. All these plots show $6.25\%$ of the test data.}
	\label{fig:reg_scatter_higher_mu}
\end{figure}

The predicted versus the true values of both observables on the \mbox{$60 \times 4$} lattice are shown in Fig.~\ref{fig:reg_scatter_higher_mu}. The individual rows correspond to the respective best model of each architecture, according to the validation loss. Although the extrapolation to higher $\lvert \phi \rvert^2$ seems to be more difficult than to higher~$n$, the predictions of the \EQ{} architecture's best model remain close to the identity line, and they are visually better than the predictions of the other two architectures' best models, the \FLAT{} model performing the worst. This leads to a visible deviation in the ensemble averages of the observables only for~\mbox{$\mu = 1.5$} and is comparable on all lattice sizes under consideration, with the exception of the \FLAT{} architecture, which allows only for predictions on the \mbox{$60 \times 4$} lattice without adapting the architecture and retraining. Note that ``best'' refers to the validation loss and that there are models of each architecture that extrapolate better than the respective ones depicted in Fig.~\ref{fig:reg_scatter_higher_mu}. However, since we are analyzing the generalization capabilities of the networks, we are restricted to metrics that take into account only the training and the validation data, and we chose the validation loss.

Figure~\ref{fig:reg_loss_over_mu} shows the total and individual test losses over~$\mu$ on the \mbox{$60 \times 4$} lattice. While the large difference between the different architectures in performance on chemical potentials smaller than \mbox{$\mu = 1.05$} is quite substantial, the performance on larger values of~$\mu$ differs by less. At \mbox{$\mu = 1.5$}, for example, the mean and median losses of the \EQ{} architecture are lower than their respective counterparts belonging to the other architectures, but there the \ST{} architecture's best model leads to the lowest. An analogous comparison between the \EQ{} and the \ST{} architecture on other lattice sizes leads to similar results, with the exception of the \mbox{$100 \times 5$} lattice, on which the latter fails. Note that Fig.~\ref{fig:reg_scatter_higher_mu} depicts the model with the lowest validation loss pertaining to each individual architecture. For the \EQ{} architecture, it is a model of the ensemble that has been trained on $20000$ training samples, whereas for the \ST{} and the \FLAT{} architecture, it is a model that has been trained on $18000$ training samples. Figure~\ref{fig:reg_loss_over_mu}, however, shows the ensemble of models trained on $20000$ training examples for each individual architecture.

\begin{figure}[tbp]
	\centering
	\includegraphics{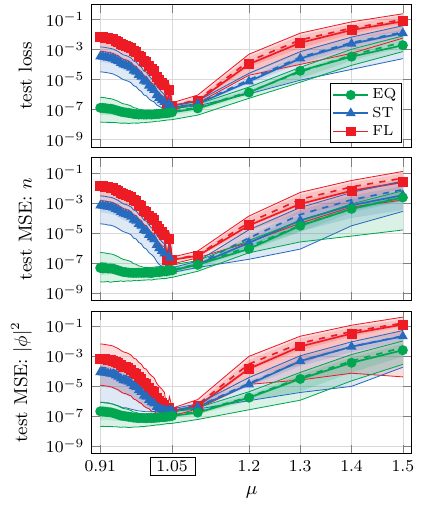}
	\caption{Total test loss and its parts corresponding to the individual observables~$n$ and $\lvert \phi \rvert^2$ over the chemical potential on the  \mbox{$60 \times 4$} lattice. It displays the ensemble of models that have been trained on $20000$ training samples corresponding to each architecture. The large difference in the quality of the predictions for \mbox{$\mu \le 1.05$} is also visible in Fig.~\ref{fig:test_loss_60_times_4}. For \mbox{$\mu > 1.05$}, the performance is more similar, although for \mbox{$\mu = 1.5$}, the mean value of the total test loss of the \ST{} models still differs from the mean values of the other architecture's models by roughly one order of magnitude.}
	\label{fig:reg_loss_over_mu}
\end{figure}

\subsubsection{Results summary}

In summary, the best translationally equivariant architecture performs better than the respective best model of the other two types on the lattice size that they have been trained on. Only some \ST{} networks are able to generalize beyond values of observables that they have encountered during training, while \EQ{} networks show no such problem. The \FLAT{} architecture shows similar behavior to the \ST{} architecture, but its predictions are less precise overall. Models of \FLAT{} architectures cannot be applied to different lattice sizes. The \EQ{} and the \ST{} architecture are both capable of generalizing to different lattice sizes, although the latter retains the higher average test loss from the \mbox{$60 \times 4$}~lattice due to the bad generalization to observable values that were not in the training set. Furthermore, \ST{} architectures are not suited to make predictions on every arbitrary lattice size. Each lattice dimension has to regard the behavior of the spatial pooling layers in the network in order to use all the data for the prediction. \EQ{} architectures have the advantage to impose no such restriction. Even though all the models have been trained only on~\mbox{$\mu = 1.05$} on the \mbox{$60 \times 4$}~lattice, many of them are able to predict the Silver Blaze phase transition on a different lattice size, where~\mbox{$\mu_c \ll 1.05$}. The \EQ{} architectures do this especially well. We found that data augmentation does not help in the training of \ST{} and \FLAT{} architectures, which is why we refrain from using it in the next two tasks.

Lastly, we want to make a comparison to the results of~\cite{Zhou:2019} where the same regression task was performed. Our best model needs much fewer trainable parameters than the one in~\cite{Zhou:2019}, i.e.~approximately \mbox{$3 \times 10^4$} compared to over \mbox{$10^7$} as extracted from their network architecture. We also found well-performing models that contain by an order of magnitude fewer parameters than our best one. Furthermore, their network architecture would fall in our \FLAT{} category, which means that it can be employed on only one specific lattice size.

\section{Classification: detecting flux violations} \label{sec:Classification}

\begin{figure*}
\subfigure[~Example field configuration]{\includegraphics{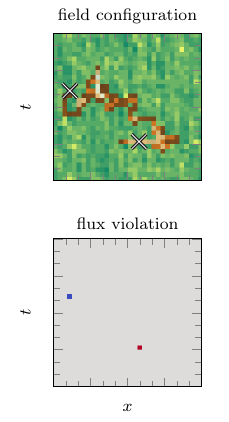}}
\subfigure[~Feature maps of convolutional network in best \EQ{} and \ST{} models]{\includegraphics{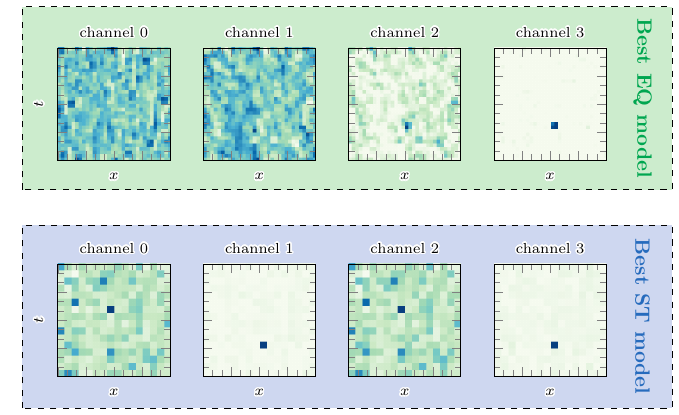}}
	\caption{Visualization of an open worm field configuration and of the best models' predictions. (a) An example field configuration including an open worm (highlighted in brown) and the resulting flux violation given by Eq.~\eqref{eq:flux_conservation}. The point-like violations occur at the two open ends of the worm (shown as crosses). (b) Feature maps of the convolutional part of the best \EQ{} (top, green) and \ST{} (bottom, blue) model showing the first four channels (of 32 and 16, respectively). Because of overparameterization, only some of the channels detect the violation (e.g.~channels 2 and 3 for \EQ{} and 1 and 3 for \ST{}), while other channels (e.g.~0 and 1 for \EQ{} and 0 and 2 for \ST{}) do not produce easily interpretable output.}
	\label{fig:openworm_schematic}
\end{figure*}

In the previous section, we have found that rather simple CNN models can approximate the functions $n$ and $|\phi|^2$ sufficiently well. In fact, the function $n$ can be exactly represented by a linear, equivariant model with a single $1 \times 1$ convolution. Similarly, while $|\phi|^2$ does not admit an exact representation in terms of $1 \times 1$ convolutions, it is easy to see that the lattice averaged quantity $\sum_x |\phi_x|^2 / (N_t N_x)$ can be written as a sum over a function that receives dominant contributions from  $k_{x,\mu}$ and $l_{x,\mu}$ at the same lattice site $x$.

In order to study models that require larger kernel sizes, we need to shift our focus to quantities that cannot be computed by taking into account only field values at a single lattice site. One example for such a quantity is the local flux violation
\begin{equation}
\mathcal{F}_x \equiv \sum_{\nu=1}^D \left( k_{x,\nu} - k_{x-\hat{\nu}, \nu} \right) \in \mathbb{Z}.
\label{eq:flux_conservation}
\end{equation}
Evaluated at some lattice site $x$, it specifically requires information from nearest neighbors surrounding $x$.

We therefore propose to solve the following classification task: an arbitrary field configuration $X = \{ k_{x,\mu}, l_{x, \nu} \}$ is mapped to the label $y(X)$:
\begin{align}
y(X) = \left\{
\begin{array}{ll}
0 &\quad \mathcal{F}_x = 0, \quad \forall x, \\
1 &\quad \rm{else}. \\
\end{array}
\right.
\end{align}
Since the worm algorithm generates only physical field configurations which by design satisfy the flux constraint $\mathcal{F}_x = 0$, $\forall x$, we adapt it to generate configurations including open worms. The field configurations generated this way exhibit flux violations at each end of the open worm (see Fig.~\ref{fig:openworm_schematic}). While we will be using such open worm configurations only for the purpose of classification and regression tasks, they are typically utilized in the calculation of $n$-point functions of $\phi$ \cite{Gattringer:2013a, Rindlisbacher:2016}.

For this task (and the following counting task in section~\ref{sec:MultipleWorms}), we have generated field configurations on square lattices of various sizes given by $(N_t \times N_x) \in \{{8 \times 8}, { 16 \times 16 }, {32\times 32}, {64\times 64} \}$. The value of the coupling constant is fixed to $\lambda = 1$, the mass $m$ takes values given by $\eta = 4 + m^2  \in \{ 4.01, 4.04, 4.25\}$, and possible values of the chemical potential $\mu$ are given by $\mu \in  \{ 1, 1.25, 1.5\}$. Training is performed only on the smallest lattice size ($8 \times 8$) and two specific choices for the pair $(\eta, \mu)$: $(\eta_1, \mu_1) = (4.25, 1)$ and $(\eta_2, \mu_2) = (4.01, 1.5)$. We use a fixed number of training examples, $N_\mathrm{train} = 4000$, distributed equally between the two classes: on half of the generated field configurations, we generate an open worm on top of a flux-constraining configuration. Other combinations of parameters and lattice sizes are used only during testing. Further details regarding the datasets can be found in appendix~\ref{app:Datasets}.

\subsection{Architecture search, training and testing}

We aim to make a comparison between the three different architecture types that have been presented in Fig.~\ref{fig:architectures}. As discussed previously, both \EQ{} and \ST{} architectures can be applied to field configurations of varying lattice size, while \FLAT{} architectures are compatible only with a fixed lattice size. As we are dealing with a binary classification problem, a sigmoid activation function is applied to the output of our models. 

To facilitate a fair comparison among architecture types, we use \textit{optuna} to perform a search for well-performing architectures using validation loss (binary cross entropy loss) as the metric to optimize for. In all three cases, we allow for up to $N_\mathrm{conv, max} = 3$ convolutional layers with circular padding and a maximum kernel size of $K = 3$ and $N_\mathrm{ch} \in \{ 4, 8, 16, 32 \}$ possible channels. Every convolution is followed by applying a \textit{LeakyReLU} activation function. In addition, after every convolution except the last we allow for the insertion of a pooling layer (either average or max pooling) with stride $s = 1$ in the case of \EQ{} networks and $s = 2$ in the case of \ST{} and \FLAT{} networks. For non-equivariant architectures we require at least one pooling layer with $s=2$ to break translational equivariance. Following this convolutional part of the network, we either apply a global max pooling layer (\EQ{} and \ST{}) or flatten the remaining lattice structure (\FLAT{}). Although other global pooling layers are possible (e.g.~average pooling or sum pooling), global max pooling seems to be the most fitting choice when the task is to detect point-like defects in the field configuration. We note that, as an additional search parameter, we allow for explicitly setting bias terms to zero in every convolutional layer.  The resulting feature map is then fed to a dense network with up to $N_\mathrm{dense, max}  = 2$ layers with $N_\mathrm{nodes} \in \{ 4, 8, 16, 32 \}$ nodes. Every linear layer is followed by the application of \textit{LeakyReLU}. A final linear layer is used to map the activation values to a single output node, which is followed by a sigmoid activation function. As before, we use binary search parameters for setting bias terms to zero in each linear layer.

For each architecture type, we perform two \textit{optuna} search runs with 400 trials each. Each model candidate (i.e.~a set of hyperparameters) is trained five times with randomly initialized weights to reduce random fluctuations from the stochastic optimization algorithm. Among the two searches, the best-performing architecture (according to validation loss) is chosen and retrained 50 times to build an ensemble of models for each architecture type. 

Training proceeds similar to the regression task of section~\ref{sec:Regression}. We use the \textit{AMSGrad} variant of the \textit{AdamW} optimizer without weight decay, a learning rate of $\lambda_{lr} = 10^{-3}$, a batch size of 100 and 200 epochs. We employ early stopping based on validation loss with a \textit{patience} value of 50. The validation set consists of 2000 examples from the same distribution as the training set. 

\begin{table}
\centering
\scriptsize
\caption{\label{tab:class_archs} Best architectures for detecting flux violations. This table shows feed-forward architectures as found by our \textit{optuna} searches. Input in the form of $(N_t, N_x, 4)$ tensors is fed into the network at the top. The output of each network is a classification probability. The last row denotes the number of trainable parameters for each type. We use an asterisk (*)  to denote layers where the bias is explicitly set to zero.}
\begin{ruledtabular}
\begin{tabular}{l l  l}
\textbf{\EQ{}} & \textbf{\ST{}}  & \textbf{\FLAT{}} \\
 \hline
Conv($2 \times 2$, 4, 32)  & Conv${}^*$($2 \times 2$, 4, 16)    & Conv${}^*$($3 \times 3$, 4, 8)    \\
LeakyReLU                  & LeakyReLU                          & LeakyReLU                         \\
Conv($1 \times 1$, 32, 32) & MaxPool($2 \times 2$, $2$)         & MaxPool($2 \times 2$, $2$)        \\
LeakyReLU                  & Conv($1 \times 1$, 16, 16)         & Conv($2 \times 2$, 8, 32)         \\
GlobalMaxPool              & LeakyReLU                          & LeakyReLU                         \\
Linear(32, 32)             & Conv($1 \times 1$, 16, 8)          & AvgPool($2 \times 2$, $2$)        \\
LeakyReLU                  & LeakyReLU                          & Conv($2 \times 2$, 32, 32)        \\
Linear${}^*$(32, 1)        & GlobalMaxPool                      & LeakyReLU                         \\
Sigmoid                    & Linear${}^*$(8, 32)                & Flatten                           \\
                           & Linear(32, 1)                      & Linear${}^*$(128, 1)              \\
                           & Sigmoid                            & Sigmoid                           \\
\hline
2657                       & 953                                & 5600                              \\
\end{tabular}
\end{ruledtabular}
\end{table}

The best architectures found during the \textit{optuna} search for each type are shown in table~\ref{tab:class_archs}.

\subsection{Results}

\begin{figure}
    \centering
    \includegraphics{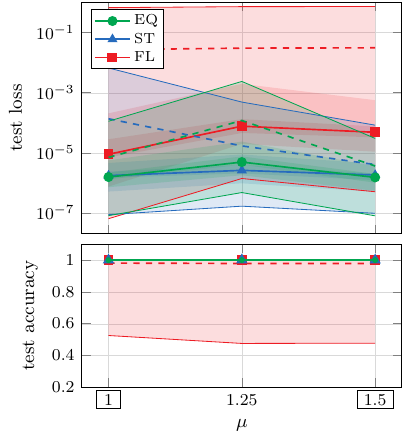}
    \caption{Top: test loss for best equivariant (\EQ{}, green), non-equivariant strided (\ST{}, blue) and non-equivariant flattening (\FLAT{}, red) classification architectures as a function of the chemical potential $\mu$ on $8 \times 8$ lattices. Training was performed on data with $\mu=1$ and $1.5$ only. Bottom: test accuracy as a function of $\mu$. The colored bands show the ensemble uncertainty from all 50 randomly initialized models with the thick line indicating the median loss (accuracy) and the dashed line showing the mean loss (accuracy). Both \EQ{} and \ST{} architectures outperform the \FLAT{} architecture.}
    \label{fig:class_res_8x8}
\end{figure}

\begin{figure}
    \centering
    \includegraphics{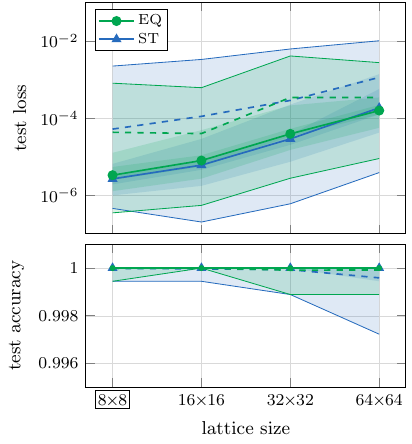}
    \caption{Top: test loss for best equivariant (\EQ{}, green) and non-equivariant strided (\ST{}, blue) classification architectures as a function of  lattice size. Bottom: test accuracy as a function of lattice size. The networks have been trained on the $8 \times 8$ lattice only. We observe that both types of architecture lead to good generalization across lattice sizes with slightly less variation in the performance of the \EQ{} architecture.}
    \label{fig:class_res}
\end{figure}

Our main results are presented in Figs.~\ref{fig:class_res_8x8} and \ref{fig:class_res}. 
Figure \ref{fig:class_res_8x8} shows a comparison of all three architecture types evaluated on $8 \times 8$ lattices as a function of $\mu$. Both \EQ{} and \ST{} exhibit very good classification accuracy and test loss, while our ensemble of  \FLAT{} models contains a few outliers which increase the average test loss. Figure \ref{fig:class_res}, which shows an average (loss and accuracy) of all available test datasets, demonstrates that both \EQ{} and \ST{} architectures generalize well on larger lattices. \FLAT{} architectures are not included, since they can be used for only one specific lattice size (in this case, $8 \times 8$). It is evident that \FLAT{} models perform worse on average compared to \EQ{} and \ST{} networks, but, in contrast to the previous regression task, the non-equivariant architecture without flattening (\ST{}) exhibits similar performance to the equivariant architecture (\EQ{}). The loss of spatial information due to pooling operations with stride $s > 1$ does not seem to affect the ability of \ST{} models to correctly classify flux violations.

In light of the results in Figs.~\ref{fig:class_res_8x8} and \ref{fig:class_res}, the question arises how \EQ{} and \ST{} models are able to make predictions with such high accuracy and if the computation that is performed by the networks can be easily understood and interpreted. To answer this, we ``dissect'' fully trained \EQ{} and \ST{} models by examining the feature maps that are generated by the convolutional part of the network, i.e.~we modify models by removing the global pooling operation and the dense network. Examples of these feature maps are shown in Fig.~\ref{fig:openworm_schematic}~(b). We see that some of the channels of the output of the convolutional part of the network highlight flux violations in the vicinity of one of the open ends of the worm (see Fig.~\ref{fig:openworm_schematic}~(a)). At first, it seems surprising that only one of the two open ends is detected. However, one has to keep in mind that the models were not directly trained on the local flux violation as in Eq.~\eqref{eq:flux_conservation} but instead were only given global information about whether or not a field configuration contains a violation. Detecting a single defect is sufficient to make the correct prediction. 

It is further noteworthy that, compared to typical deep learning models, the models found in our architecture search are rather small, with $\sim 2700$ parameters in the case of \EQ{} models and $\sim 1000$ parameters for our best \ST{} architecture on this task.

\section{Regression: counting flux violations}
\label{sec:MultipleWorms}

A natural extension of last section is the study of lattice configurations with more than one open worm, meaning a regression task where the inputs are lattice configurations that are labeled by the number of open worms $N_\mathrm{worms}$ that they contain. The addition of an open worm implies the emergence of a flux violation at its head and tail, meaning that the quantity defined in Eq.~\eqref{eq:flux_conservation} respects $\mathcal{F}_x=\pm1$ at an open worm endpoint. As discussed in more detail in Appendix~\ref{app:Datasets}, we explicitly forbid endpoints of different worms to lie on top of each other; therefore, a configuration with $N_\mathrm{worms}$ open worms is characterized by $2N_\mathrm{worms}$ points where $\mathcal{F}_x=\pm1$. With this clarification, the task we are going to tackle in this section can be formally expressed as the approximation of the function
\begin{equation}
    y(X)=\displaystyle\frac{1}{2}\sum_x\left|\mathcal{F}_x\right|,
\end{equation}
where $X$ is a lattice configuration $\{k_\mu,l_\nu\}$. We note that this task resembles a simplified version of other counting problems, such as crowd counting \cite{Gao:2020}.

The physical parameters are the ones mentioned in the previous section, with the addition of a number of open worms ranging from 0 to 10, yielding a total of $36\times11=396$ combinations of parameters. While the test set includes data coming from all these combinations, the training set consists of data created at only a small subset of such combinations to inspect the generalization capabilities of the architecture under consideration. We use a training set with $N_\mathrm{train}=20000$ samples distributed equally between two different numbers of open worms $N_\mathrm{worms}\in\{0,5\}$ and physical parameters $(\eta,\mu)\in\{(4.01,1.5),(4.25,1)\}$. The validation set contains $N_\mathrm{val}=2000$ samples. For more details regarding the datasets, see appendix~\ref{app:Datasets}. 

\subsection{Architecture search, training and testing}

A preliminary phase is carried out in order to explore trends with different hyperparameter choices. We also empirically confirm the relationship between the prediction of an extensive quantity and the necessity of a global sum after the convolutional part of the neural network, as discussed in section~\ref{sec:translational_symmetry}. The information gathered in this initial stage is then used to determine the architecture search space for \textit{optuna}. As in the two previous tasks, this is done for the three architecture types shown in Fig.~\ref{fig:architectures}. The search spaces are designed to be as similar as possible to eliminate favorable conditions for any of the three architecture types.

The \EQ{} architecture search space is characterized by $N_\mathrm{conv}\in\{2,3,4\}$ convolutional layers with a kernel size $K\in\{1,2,3\}$, followed by a global sum pooling layer which leads to a dense network, composed of $N_\mathrm{dense}\in\{0,1,2\}$ layers. The \ST{} architecture search space is structured in the same way with the additional insertion of $N_\mathrm{pool}\in\{1,2\}$ spatial pooling layers with stride \mbox{$s=2$}. Since training is conducted on $8\times8$ lattices, three such pooling layers would reduce the lattice to only one site and render global sum pooling ineffective, which is why we limit the choice of $N_{\mathrm{pool}}$. The \FLAT{} architecture search space features two mandatory convolutions with a $2\times2$ or a $3\times3$ kernel, each followed by a spatial pooling layer. A $1\times1$ convolution can be inserted before and after each mandatory convolution, leading to a total number of convolutions $N_\mathrm{conv}^{'}\in\{2,3,4,5,6\}$. This part is followed by the flattening layer and a dense network consisting of $N_\mathrm{dense}^{'}\in\{1,2,3\}$ layers, where the maximum number of layers is increased with respect to the other two architecture types to compensate for the possible absence of $1\times1$ convolutions. All three types share the following features: circular padding is used in every convolution; the channels in the convolutions and the nodes in the dense layers are selected from the set $N_\mathrm{ch/nodes}\in\{4,8,16,32\}$; a \textit{LeakyReLU} activation function is used after every convolution and every linear layer not leading to the output; the bias in both the convolutions and the linear layers is turned off. We also mention that an independent search is run also for \EQ{} architectures with the optional inclusion of spatial pooling layers with stride $s=1$, in the same fashion described for \ST{} models. However, none of the \EQ{} models found in this run are better than the \EQ{} models found in the previous search.

\begin{table}
\centering
\scriptsize
\caption{\label{tab:reg2_archs} Best architectures for counting flux violations. This table lists the feed-forward architectures resulting from the \textit{optuna} searches sorted by their average validation loss over 20 instances trained from scratch. Four channels of size $N_t\times N_x$ are the input tensors passed at the top of each network, which yields a scalar output representing the predicted number of open worms. The last row shows the number of trainable parameters for each architecture.}
\begin{ruledtabular}
\begin{tabular}{l l l}
\textbf{1st \EQ{}} & \textbf{2nd \EQ{}} & \textbf{3rd \EQ{}} \\
 \hline
Conv($1 \times 1$, 4, 32) & Conv($2 \times 2$, 4, 8) & Conv($1 \times 1$, 4, 4) \\
LeakyReLU & LeakyReLU & LeakyReLU \\
Conv($2 \times 2$, 32, 8) & Conv($2 \times 2$, 8, 8) & Conv($2 \times 2$, 4, 8) \\
LeakyReLU & LeakyReLU & LeakyReLU \\
Conv($2 \times 2$, 8, 16) & Conv($1 \times 1$, 8, 4) & Conv($2 \times 2$, 8, 4) \\
LeakyReLU & LeakyReLU & LeakyReLU \\
Conv($1 \times 1$, 16, 8) & Conv($1 \times 1$, 4, 8) & Conv($3 \times 3$, 4, 1) \\
LeakyReLU & LeakyReLU & LeakyReLU \\
GlobalSumPool & GlobalSumPool & GlobalSumPool \\
Linear(8, 1) & Linear(8, 1) & \\
\hline
1800 & 456 & 308 \\
\end{tabular}
\medskip

\begin{tabular}{l l l}
\textbf{1st \ST{}} & \textbf{2nd \ST{}} & \textbf{3rd \ST{}} \\
 \hline
Conv($2 \times 2$, 4, 16) & Conv($2 \times 2$, 4, 4) & Conv($2 \times 2$, 4, 4) \\
LeakyReLU & LeakyReLU & LeakyReLU \\
Conv($1 \times 1$, 16, 32) & MaxPool($2 \times 2$, $2$) & AvgPool($2 \times 2$, $2$) \\
LeakyReLU & Conv($2 \times 2$, 4, 4) & Conv($3 \times 3$, 4, 16) \\
Conv($1 \times 1$, 32, 32) & LeakyReLU & LeakyReLU \\
LeakyReLU & GlobalSumPool & GlobalSumPool \\
AvgPool($2 \times 2$, $2$) & Linear(4, 1) & Linear(16, 32) \\
Conv($1 \times 1$, 32, 8) & & LeakyReLU \\
LeakyReLU & & Linear(32, 1) \\
GlobalSumPool & & \\
Linear(8, 32) & & \\
LeakyReLU & & \\
Linear(32, 1) & & \\
\hline
2336 & 132 & 1184 \\
\end{tabular}
\medskip

\centering
\scriptsize
\begin{tabular}{l l l}
\textbf{1st \FLAT{}} & \textbf{2nd \FLAT{}} & \textbf{3rd \FLAT{}} \\
 \hline
Conv($2 \times 2$, 4, 4) & Conv($2 \times 2$, 4, 8) & Conv($2 \times 2$, 4, 32) \\
LeakyReLU & LeakyReLU & LeakyReLU \\
AvgPool($2 \times 2$, $2$) & AvgPool($2 \times 2$, $2$) & AvgPool($2 \times 2$, $2$) \\
Conv($3 \times 3$, 4, 8) & Conv($3 \times 3$, 8, 4) & Conv($3 \times 3$, 32, 4) \\
LeakyReLU & LeakyReLU & LeakyReLU \\
AvgPool($2 \times 2$, $2$) & AvgPool($2 \times 2$, $2$) & AvgPool($2 \times 2$, $2$)
\\
Flattening & Flattening & Flattening 
\\
Linear(8, 4) & Linear(4, 4) & Linear(4, 32)
\\
LeakyReLU & LeakyReLU & LeakyReLU \\
Linear(4, 32) & Linear(4, 32) & Linear(32, 16) \\
LeakyReLU & LeakyReLU & LeakyReLU \\
Linear(32, 1) & Linear(32, 1) & Linear(16, 1) \\
\hline
640 & 640 & 2704 \\
\end{tabular}
\end{ruledtabular}
\end{table}

As in the previous section, two metrics are employed for performance analysis: the MSE loss and the accuracy, for which predictions are rounded to the closest integer. The quantity monitored during the optimization phase is validation loss. Since the hyperparameter search spaces are large, two \textit{optuna} runs are executed to reduce the risk of overlooking promising regions. For each hyperparameter selection, three models are trained, in order to attenuate initialization influences, for 200 epochs with no early stopping. The other hyperparameters are defined prior to the optimization: we adopt a batch size of 16, a learning rate $\lambda_{lr}=10^{-3}$ and the \textit{AMSGrad} variant of the \textit{AdamW} optimizer with zero weight decay.

Out of 100 different architectures from the two \textit{optuna} searches, the best three for each type are selected according to the validation loss averaged over their three initializations. These architectures become the starting point of the next step: training the most promising architectures from scratch.

We keep all the same hyperparameters, except for the number of epochs which is increased to 500, and the same training and validation sets. For a fair comparison 20 instances of the same architectures are trained to mitigate the influence of random initializations, and for each of them the best model is saved. We sort the architectures according to the average over the 20 models of the validation loss. Table~\ref{tab:reg2_archs} portrays the details of the feed-forward networks.

\subsection{Results}

\begin{figure}
    \centering
    \includegraphics{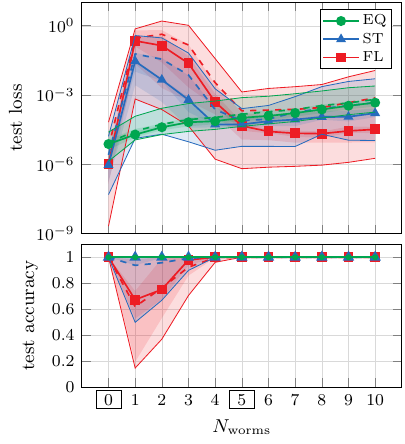}
    \caption{Test loss (top) and test accuracy (bottom) of the best architectures according to the mean of the validation loss tested on all $8\times8$ lattices
    as functions of the number of open worms. Training and validation are carried out at $N_\mathrm{worms} = 0$ and $N_\mathrm{worms} = 5$, while test results are shown for $N_\mathrm{worms} \in [0, 10]$.}
    \label{fig:reg2_val_mean_res_8x8}
\end{figure}

\begin{figure}
    \centering
    \includegraphics{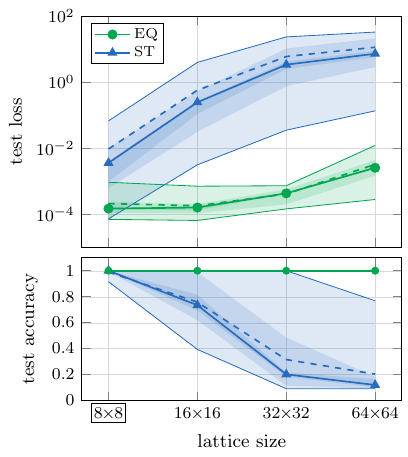}
    \caption{Test loss (top) and test accuracy (bottom) of the best architectures according to the mean of the validation loss tested as functions of the lattice size. Training and validation are carried out on the smallest lattice ($8 \times 8$), while testing is performed on all lattice sizes.}
    \label{fig:reg2_val_mean_res}
\end{figure}

Since \FLAT{} models cannot be evaluated on input sizes different from the ones they have been trained on, we make two kinds of comparison between architectures: one involves all three types tested only on $8\times8$ lattices, and the other focuses on \EQ{} and \ST{} tested on all lattice sizes available. The first analysis is featured in Fig.~\ref{fig:reg2_val_mean_res_8x8} and the second in Fig.~\ref{fig:reg2_val_mean_res}, where the dashed lines indicate the mean values and the markers represent the medians. 

A common takeaway of these plots is that for this task equivariance proves to be an important property to incorporate into the network. Interestingly, Fig.~\ref{fig:reg2_val_mean_res_8x8} suggests that \ST{} and even more so \FLAT{} have difficulties in recognizing certain numbers of open worms, with the lowest performance at $N_\mathrm{worms}=1$, which is compatible with the fact that the training set consists only of $N_\mathrm{worms}\in\{0,5\}$. 

\begin{table*}
\centering
\scriptsize
\caption{\label{tab:reg2_metrics} Metrics of the best architectures for counting flux violations. Highlighted in bold are the results of the best architectures for each type according to the corresponding metric.}
\begin{ruledtabular}
\begin{tabular}{l c c c c c c}
 & \multicolumn{2}{c}{\textbf{validation loss on $8\times8$}} & \multicolumn{2}{c}{\textbf{test loss on $8\times8$}} & \multicolumn{2}{c}{\textbf{test loss up to $64\times64$}} \\
& \textbf{mean} & \textbf{median} & \textbf{mean} & \textbf{median} & \textbf{mean} & \textbf{median} \\
\hline \\[-2ex]
\textbf{1st \EQ{}} & $\pmb{\num{4.676e-05}}$ & $\num{4.137e-05}$ & $\pmb{\num{2.108e-04}}$ & $\num{1.483e-04}$ & $\pmb{\num{1.008e-03}}$ & $\num{8.308e-04}$ \\
\textbf{2nd \EQ{}} & $\num{1.042e-04}$ & $\pmb{\num{2.440e-05}}$ & $\num{3.525e-04}$ & $\pmb{\num{8.783e-05}}$ & $\num{1.807e-03}$ & $\pmb{\num{7.936e-04}}$ \\
\textbf{3rd \EQ{}} & $\num{8.992e-03}$ & $\num{3.072e-04}$ & $\num{2.105e-02}$ & $\num{9.163e-04}$ & $\num{1.925e+00}$ & $\num{4.031e-02}$ \\
\hline \\[-2ex]
\textbf{1st \ST{}} & $\pmb{\num{2.331e-05}}$ & $\num{2.173e-05}$ & $\num{9.438e-03}$ & $\num{3.576e-03}$ & $\num{4.446e+00}$ & $\num{3.026e+00}$ \\
\textbf{2nd \ST{}} & $\num{8.479e-05}$ & $\num{4.372e-05}$ & $\pmb{\num{2.545e-04}}$ & $\pmb{\num{9.340e-05}}$ & $\pmb{\num{3.738e-03}}$ & $\pmb{\num{1.171e-03}}$ \\
\textbf{3rd \ST{}} & $\num{2.869e-04}$ & $\pmb{\num{2.171e-05}}$ & $\num{1.676e-02}$ & $\num{1.381e-03}$ & $\num{2.943e+00}$ & $\num{9.580e-01}$ \\
\hline \\[-2ex]
\textbf{1st \FLAT{}} & $\pmb{\num{2.602e-05}}$ & $\num{1.787e-05}$ & $\num{7.837e-02}$ & $\num{3.817e-02}$ & - & - \\
\textbf{2nd \FLAT{}} & $\num{4.004e-05}$ & $\num{1.117e-05}$ & $\pmb{\num{5.300e-02}}$ & $\pmb{\num{1.285e-03}}$ & - & - \\
\textbf{3rd \FLAT{}} & $\num{5.805e-05}$ & $\pmb{\num{1.031e-05}}$ & $\num{6.382e-02}$ & $\num{3.556e-02}$ & - & - \\
\end{tabular}
\end{ruledtabular}
\end{table*}

We observe that the podium ordering depends on the metric chosen; for example, in table~\ref{tab:reg2_metrics}, the mean and the median of the validation loss lead to different winners for all architecture types. As the training and validation sets are characterized by the same physical parameters, which represent a very small subset of the whole set of parameters used for testing, metrics on the validation set may not be indicative of the generalization capabilities of the network, so we investigate the same metrics on the test set in the two manners described at the beginning of this subsection.

\begin{figure}
    \centering
    \includegraphics{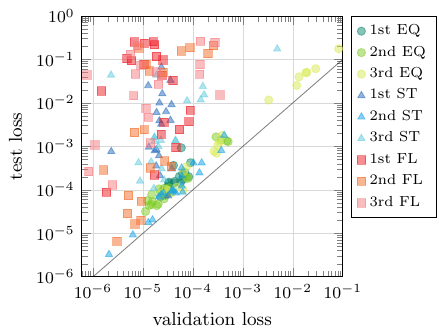}
    \caption{Test loss on $8\times8$ lattices versus validation loss of every instance for each architecture. This scatter plot shows 20 models obtained during retraining for the three winning architectures of each type (\EQ{}, \ST{} and \FLAT{}). The diagonal black line indicates where validation loss equals test loss. Networks have been trained and validated for $N_\mathrm{worms}\in\{0,5\}$ and $(\eta,\mu)\in\{(4.01,1.5),(4.25,1)\}$ on an $8 \times 8$ lattice. Generalization (test loss) is checked with zero to ten open worms, $\mu \in  \{1.0, 1.25, 1.5 \}$, $\eta \in \{ 4.01, 4.04, 4.25\}$ and a fixed lattice size of $8 \times 8$. The closer a particular point lies to the black line, the better it generalizes. This appears to be generally the case for \EQ{} architectures (green circles).}
    \label{fig:reg2_tloss_vloss_8x8}
\end{figure}

\begin{figure}
    \centering
    \includegraphics{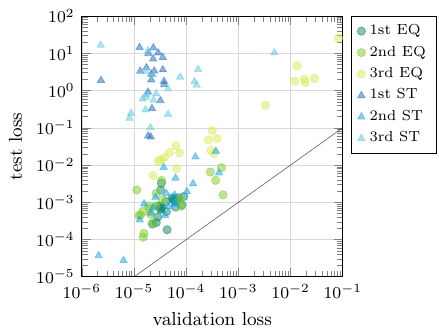}
    \caption{Scatter plot of test loss on all lattice sizes versus validation loss of every instance for \EQ{} and \ST{} architectures. Similar to Fig.~\ref{fig:reg2_tloss_vloss_8x8}, we demonstrate the generalization capabilities of our models to different lattice sizes from $8 \times 8$ up to $64 \times 64$ and different physical parameters, while being trained on only $8 \times 8$. In particular, \EQ{} models (green circles) are closer to the black line where test loss and validation loss agree.}
    \label{fig:reg2_tloss_vloss}
\end{figure}

An in-depth analysis is shown in Figs.~\ref{fig:reg2_tloss_vloss_8x8} and~\ref{fig:reg2_tloss_vloss}, depicting the relationship between the test and validation loss for all 20 models of each architecture. If an architecture is prone to generalization issues, its instances are scattered mostly vertically, which manifestly happens to most of the \ST{} and \FLAT{} architectures. \EQ{} models are instead distributed closer to the black line, where test and validation loss are equal. These results suggest that for \EQ{} architectures low validation loss correlates with low test loss, and therefore they tend to reliably generalize.

Remarkably, there is also a non-equivariant architecture featuring this behavior, specifically the 2nd \ST{}, whose best two instances even outperform the best \EQ{} models by almost an order of magnitude both in the validation and in the test loss. This is an illustration that the validation procedure does not guarantee generalization if the validation set is restricted to a small set of physical parameters. Indeed, the test loss on $8\times8$ lattices in the central column in table~\ref{tab:reg2_metrics} already contains a generalization in terms of physical parameters, which implies that it would be sufficient to include at least some of those configurations in the validation set in order to select the best generalizing architecture on different lattice sizes.

We observe two distinctive properties of this outstanding \ST{} architecture, that are possible contributing factors to its success: it is the only one containing a spatial max pooling layer and it is characterized by a very small number of parameters, the smallest of all the examined architectures, as can be seen in table~\ref{tab:reg2_archs}. Max pooling can be beneficial for the detection of local defects, as pointed out in the previous task, while the relative simplicity of this counting problem calls for simpler network structures. Indeed, \textit{optuna} favors overall small architectures for this task too, with a number of parameters between $\sim 100$ and $\sim 3000$ for the ones studied in detail.

\section{Conclusions and outlook} \label{sec:ConclusionsAndOutlook}

In this work, we studied the effect of imposing global translational invariance on convolutional neural network architectures. We did so by comparing three different architecture types that are commonly used and which differ with regards to their equivariance and generalization properties. Network architectures that use only convolutions or pooling operations of stride one and a global pooling layer before a subsequent dense network preserve translational equivariance. Such networks are also able to generalize to different input sizes if the global pooling operation is compatible with the intensive or extensive property of the output quantity. Network architectures that contain pooling operations with a stride greater than one generally break translational equivariance. Using a flattening operation instead of global pooling further impairs translational equivariance of the network and restricts its usage to one particular input shape, preventing a straightforward generalization to other lattice sizes without retraining. This latter architecture type has been particularly popular in image classification tasks and has subsequently also been used in physics applications.

We chose three different tasks related to characterizing complex scalar field configurations on a two-dimensional lattice with periodic boundary conditions that are given in the flux representation. This representation contains integer-valued field configurations which have to obey a flux conservation law. Valid configurations are generated using the worm algorithm. The first task we performed was a regression task to predict the particle density $n$ and the field average of $|\phi|^2$, given just the plain field configurations in flux quantities. The predicted observables also depend on various physical parameters, including the chemical potential~$\mu$, that are set during the generation of the configurations. We found that it is sufficient to train at only one value of~$\mu$ in order to be able to generalize to other values of~$\mu$, in particular also to extrapolate beyond the Silver Blaze phase transition. While this result seems surprising at first, it can be explained by the fact that different input configurations at fixed training chemical potential~$\mu_\mathrm{train}$ already cover a wide range of possible input values that are shared between physical parameters, i.e.~other values of~$\mu$. Comparing the three architecture types by selecting their best-performing representatives from a network architecture search using \textit{optuna}, we generally find that equivariant architectures perform best at this task. They excel, in particular, when increasing the size of the training set. We also explored whether data augmentation on the input side can compensate for missing equivariance, but this turned out to have a barely noticeable effect on the result. Furthermore, we investigated the generalization properties to smaller and larger lattice sizes, which is possible only for architectures that contain a global pooling layer. Again, across all lattice sizes, the equivariant architecture wins. Strided architectures can generalize only to lattice sizes that are multiples of the stride combinations used in the network, whereas flattening architectures are not able to generalize at all to other lattice sizes. These architectures would have to be retrained for each input size separately.

The next two tasks were related to detecting and counting flux violations from open worm configurations. Such configurations can appear in the calculation of $n$-point functions. They are particularly interesting as the result cannot be approximated by a purely local function, which would involve only $1 \times 1$ convolutions, but requires at least a $2 \times 2$ convolution. In the case of detecting flux violations, flattening models perform worst, while equivariant and strided architectures with global pooling layers are both able to predict the result with comparably high accuracy. Inspecting the feature maps of the trained models, we found that these models learn to detect only one end of the open worms, but this is sufficient to solve this task. For the third task to count the number of flux violations, we trained only on configurations containing either zero or five open worms and tested on configurations that contained any number of worms from zero to ten. In this setup, again the equivariant architecture wins compared to strided or flattening architectures. Interestingly, the networks have most problems to differentiate between zero and one worms, while a larger number of worms poses fewer problems. Another interesting observation is that the selection procedure of the network architecture search can lead to different optimal choices with very different generalization properties. Because of our particular choice of validation and test data, the validation loss alone is not sufficient to select the architecture that can generalize best to other physical parameters or network sizes. The optimal models we found were much smaller regarding the number of weight parameters than models used in comparable studies in the literature.

Based on our findings, we can clearly recommend using global pooling layers in future machine learning tasks that involve systems with global translational invariance. Global pooling layers allow to easily generalize results to different lattice sizes in regression and classification tasks. Whether the advantages of using pooling layers with a stride greater than one outweigh the possible disadvantages of breaking translational equivariance depends on the system being studied. An interesting aspect that warrants further study is the question of why some architectures seem to generalize better than others and whether there is a way to identify or characterize such architectures already before testing on an extended test set. Moreover, physical parameters may not be the best quantities for assessing generalization capabilities, but one should rather study the distribution of input and output values of the network. While in this work we concentrated on translational symmetry, it would be interesting to extend this study in the future to further symmetries on the lattice using, for example, G-CNNs. One could also examine coset pooling at intermediate layers that respect translational invariance. Finally, based on our findings, it seems worthwhile to investigate and study possible translationally equivariant versions of current architectures that explicitly break translational invariance, even though the underlying theory would respect this symmetry.

\section{Acknowledgement}

We thank Kai Zhou for correspondence. DM thanks Jimmy Aronsson for valuable discussions regarding group equivariant neural networks. This work has been supported by the Austrian Science Fund FWF No.~P32446-N27, No.~P28352 and Doctoral program No. W1252-N27. The Titan\,V GPU used for this research was donated by the NVIDIA Corporation.

\appendix
\section{The complex scalar field} \label{app:ComplexScalarField}

The action of a $1 + 1$-dimensional complex scalar field~$\phi$ in the continuum with quartic interaction, a non-zero chemical potential~$\mu$ and no external sources can be written as
\begin{equation}
    S  = \mspace{-5mu} \int \mspace{-4mu} \mathrm{d}x_0 \mathrm{d}x_1 \left(  \lvert D_0 \phi \rvert^2 - \lvert \partial_1 \phi \rvert^2  - m^2 \lvert \phi \rvert^2 - \lambda \lvert \phi \rvert^4 \right) \mspace{-2mu}, \label{action_complex_scalar_field_continuum}
\end{equation}
with $D_0 = \partial_0 - i \mu$, the mass~$m$ and the coupling constant~$\lambda$. The invariance property of the action under translations in time and space gives rise to the conservation of the energy momentum tensor. After a Wick rotation
\begin{equation}
x_0 \rightarrow i x_2,\quad x_2 \in \mathbb{R},
\end{equation}
we obtain the imaginary time version of the action in Eq.~\eqref{action_complex_scalar_field_continuum}, namely
\begin{equation}
    S_E \mspace{-1mu} = \mspace{-5mu} \int \mspace{-4mu} \mathrm{d}x_1 \mspace{-1mu} \mathrm{d}x_2 \mspace{-2mu} \left( \lvert \partial_1 \phi \rvert^2  \mspace{-2mu} + \lvert D_2 \phi \rvert^2 \mspace{-2mu} + m^2 \lvert \phi \rvert^2 \mspace{-2mu} + \lambda \lvert \phi \rvert^4 \right) \mspace{-2mu}, \label{action_complex_scalar_field_continuum_euclidean}
\end{equation}
with $D_2 = \partial_2 + \mu$ and the imaginary time~$x_2$.

The Euclidean action, which is given by Eq.~\eqref{action_complex_scalar_field_continuum_euclidean}, can be discretized, which makes it possible to analyze the complex scalar field on the lattice. The result reads (see e.g.~\cite{Gattringer:2013b})
\begin{IEEEeqnarray}{rCl}
    \nonumber S_{lat} &=& \sum_x \left( \eta \lvert \phi_x \rvert^2 + \lambda \lvert \phi_x \rvert^4 \phantom{\sum_{\nu = 1}^2} \right. \\
    && - \> \left. \sum_{\nu = 1}^2 \left( e^{\mu \mspace{2mu} \delta_{\nu, 2}} \phi_x^* \phi_{x + \hat{\nu}} + e^{- \mu \mspace{2mu} \delta_{\nu, 2}} \phi_x^* \phi_{x - \hat{\nu}} \right) \right), \IEEEeqnarraynumspace \label{action_complex_scalar_field_lattice}
\end{IEEEeqnarray}
where \mbox{$\eta = 2 D + m^2 = 4 + m^2$}, and $\delta_{\nu,2}$ is the Kronecker delta. The first sum is over all lattice sites~$x$, and the second one is over the two directions: space and imaginary time. The position~\mbox{$x + \hat{\nu}$} is reached by moving one unit vector~$\hat{\nu}$ from~$x$ in the $\nu$~direction. Naturally, periodic boundary conditions are employed. In Eq.~\eqref{action_complex_scalar_field_lattice} we have explicitly set the lattice spacing to unity. This implies that all dimensionful quantities such as $m$ and $\mu$ are understood to be given in appropriate units of the lattice spacing. We limit the extension of the system to $L$ in the spatial direction and to $1/T$ in the temporal one, where $T$ denotes the temperature.

For non-zero chemical potential~$\mu$, the action in Eq.~\eqref{action_complex_scalar_field_lattice} becomes complex. This is problematic because in this case the term~\mbox{$e^{-S_{lat}}$} cannot be interpreted as a probability distribution, and therefore it is not possible to use standard Monte Carlo sampling to determine the partition function
\begin{equation}
    Z = \int \! \mathcal{D}\phi \, e^{-S_{lat}}
\end{equation}
and its derivatives. To circumvent this so-called complex action problem, which is also known as the sign problem, one can work in a dual formulation, known as flux representation. The derivation of the partition function in the flux representation can be found in~\cite{Gattringer:2013b}. The result reads
\begin{IEEEeqnarray}{rCl}
    \nonumber Z &=& \sum_{ \{ k, l \} } \left( \prod_{x, \nu} \frac{1}{(\lvert k_{x, \nu} \rvert + l_{x, \nu})! l_{x, \nu}!} \right) \left( \prod_x e^{\mu k_{x,2}} W(f_x) \right) \\
    && \> \times \left( \prod_x \delta \left( \sum_\nu (k_{x, \nu} - k_{x-\hat{\nu}, \nu}) \right) \right), \IEEEeqnarraynumspace \label{compl:flux_representation}
\end{IEEEeqnarray}
with
\begin{IEEEeqnarray}{rCl}
    \nonumber \sum_{ \{ k, l \} } &=& \prod_{x, \nu} \mspace{8mu}\sum_{k_{x, \nu} = - \infty}^\infty \mspace{8mu} \sum_{l_{x, \nu} = 0}^\infty \\ &=&\sum_{k_{1,1}=-\infty}^\infty  \mspace{8mu} \sum_{l_{1,1}=0}^\infty \mspace{8mu} \sum_{k_{1,2}=-\infty}^\infty \cdots \sum_{l_{N,2}=0}^\infty,
\end{IEEEeqnarray}
where the $N$ lattice sites have been labeled with numbers $x\in\{1,2,\ldots,N\}$. The degrees of freedom are the four integer fields~$k_{x, \nu}$ and~$l_{x, \nu}$, where~$\nu = 1,2$. The former must obey the flux conservation law
\begin{equation}
    \sum_{\nu} \left( k_{x, \nu} - k_{x - \hat{\nu}, \nu} \right) = 0 \label{compl:flux_conservation}
\end{equation}
at all lattice sites~$x$ for the Kronecker delta not to vanish; the latter are non-negative. The function~$W(f_x)$ is given by
\begin{equation}
    W(f_x)=\int_0^\infty \mathrm{d}x\, x^{f_x+1}\mathrm{e}^{-\eta x^2-\lambda x^4}, \label{compl:W}
\end{equation}
and its integer valued argument reads
\begin{equation}
f_x=\sum_\nu[|k_{x,\nu}|+|k_{x-\hat{\nu},\nu}|+2(l_{x,\nu}+l_{x-\hat{\nu},\nu})]. \label{f_x_definition}
\end{equation}

Observables can be derived from the partition function and written in terms of the dual variables~$k_{x, \nu}$ and~$l_{x, \nu}$. In this paper, two quantities are of special interest, namely the particle number density~$n$ and the lattice averaged squared absolute value of the field~$\lvert \phi \rvert^2$. Their ensemble averages \mbox{$\langle \cdots \rangle$} are given by
\begin{IEEEeqnarray}{rCcCl}
    \langle n \rangle &=& \frac{T}{V} \frac{\partial \ln Z}{\partial \mu} &=& \frac{1}{N_x N_t} \left \langle \sum_x k_{x, 2} \right \rangle, \label{compl:n}\\
    \langle \lvert \phi \rvert^2 \rangle &=& - \frac{T}{V} \frac{\partial \ln Z}{\partial \eta} &=& \frac{1}{N_x N_t} \left \langle \sum_x \frac{W(f_x + 2)}{W(f_x)} \right \rangle,\label{compl:phi} \IEEEeqnarraynumspace
\end{IEEEeqnarray}
where $N_x$ ($N_t$) is the number of lattice sizes in the spatial (temporal) direction.

For our machine learning tasks, we associate each individual configuration $\{ k_{x,\mu}, l_{x,\mu} \}$ with particular values of $n$ and $|\phi|^2$ in Eqs.~\eqref{Regression:n_formula} and~\eqref{Regression:phi2_formula}, even though the dual formulation does not allow for a direct mapping between field configurations ${\phi_x}$ and link configurations $\{ k_{x,\mu}, l_{x,\mu} \}$.

\section{Datasets} \label{app:Datasets}

In this appendix, we discuss the Monte Carlo procedure we use to generate the datasets for our machine learning tasks.

The flux representation, which is given by Eq.~\eqref{compl:flux_representation}, is characterized by the positive field $l$ and the field $k$ constrained by Eq.~\eqref{compl:flux_conservation}. Since they are different in nature, a suitable algorithm is composed of two distinct parts, each of which takes care of the modifications of the respective field. The link variables $l$ are updated using a standard Monte Carlo algorithm, where the Metropolis acceptance probabilities are ratios of Boltzmann weights of the dual action. The links $k$ are updated by means of the worm algorithm, originally proposed in \cite{Prokofev:2001}, where the acceptance probabilities follow the prescriptions given in \cite{Gattringer:2013b}. Using these algorithms, we generate all datasets in this work.

The initial configuration is set to zero at every lattice site for both $k$ and $l$. Before reaching equilibrium, the system undergoes a thermalization phase, which we discard. Since autocorrelation in the dataset can affect the learning process, we monitor it and set an appropriate number of waiting sweeps between each measurement.

\subsection{Regression: predicting observables on the lattice}

The dataset contains lattice configurations and corresponding $n$ and $|\phi|^2$ values, the first ones being the input for the CNN and the latter being the quantities to predict. We create data with the following set of physical parameters: $\eta=4.01$, $\lambda=1$ and $\mu \in \{ 0.91,\dots,1.5 \}$, where values in the range $[0.91,1.05]$ are separated by $\Delta\mu=0.005$, while $\Delta\mu=0.1$ in the range $[1.1,1.5]$. We choose five different lattice sizes: $50\times2$, $60\times4$, $100\times5$, $125\times8$ and $200\times10$, where the first number is $N_t=1/T$ and the second one $N_x=L$. Different $N_t$ means different temperature $T$, which influences the properties of the phase transition, as shown in Fig.~\ref{fig:test_plots_mu_60_times_4}. The total amount of training data is $N_{\mathrm{train}}=20000$, generated at $\mu=1.05$ on the $60\times4$ lattice, and the whole validation set consists of $N_{\mathrm{val}}=2000$ at the same $\mu$ and lattice size. We define two distinct test sets, both containing 4000 data points per each $\mu$ at each lattice size. The first test set (test set~A) is characterized by values of $\mu \in \{ 0.91,\dots,1.05 \}$, which correspond to the ones used in~\cite{Zhou:2019}, in order that a direct comparison with the results found there is possible. The second test set (test set~B) is designed to examine the extrapolation abilities of the neural networks to chemical potentials higher than the one they have been trained on, specifically in the range $\mu \in \{ 1.1,\dots,1.5 \}$. The total amount of test data is $N_{\mathrm{test}}=4000\times5\times(29+5)=680000$. We discard the first 1000 sweeps to disregard thermalization and then save a configuration and the respective observables every five sweeps. For some combinations of chemical potential and lattice size, we observe a high autocorrelation. In these cases, the number of sweeps is increased to 50, which sufficiently reduces the autocorrelation.

We now closely inspect the distribution of the two fundamental quantities needed for the computation of the observables $n$ and $|\phi|^2$, namely $k_{x,t}$ and $f_x$. This is meant to give an additional insight into the dataset properties and the generalization capabilities of the architectures. We remind that $k_{x,t}$ can take any integer value, while $f_x$ is either 0 or a positive even number. In the following discussion, we omit the lattice index $x$ and use $k_t$, $k_x$ and $f$ instead. 

\begin{figure}
    \centering
    \includegraphics{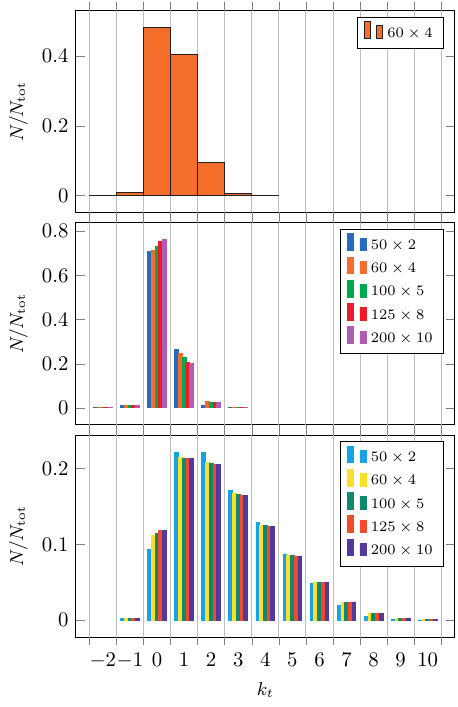}
    \caption{Distributions of the link field $k_t$. These histograms feature the distributions of $k_t$ in the training set (top), test set~A (middle) and test set~B (bottom). The test sets maintain a similar distribution along different lattice sizes. Even though training and test set~A cover the same domains, their distributions are different, which is the origin of the generalization issues of some architectures. The distribution of test set~B also reaches higher values of $k_t$, which can make a generalization to data in test set~B even more difficult than to data in test set~A. Bars corresponding to weights smaller than $10^{-4}$ in each plot are not shown.}
    \label{fig:data_distr_reg_kt}
\end{figure}

\begin{figure}
    \centering
    \includegraphics{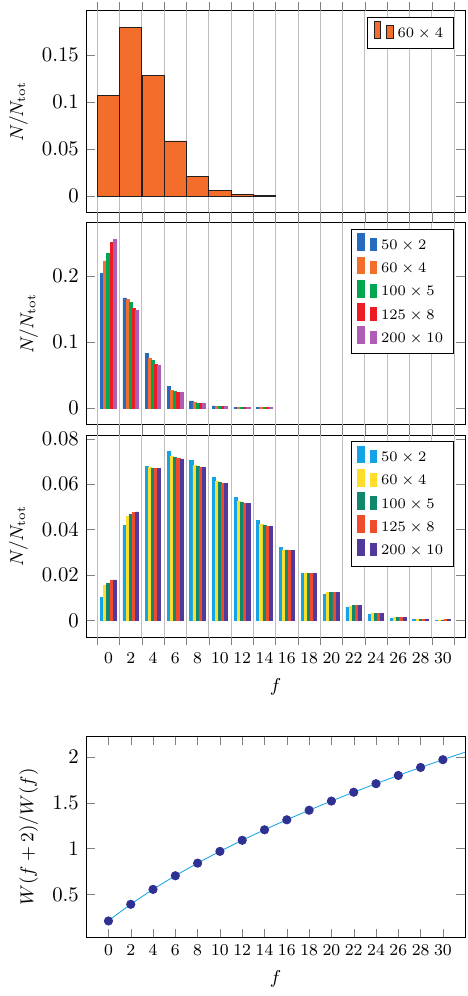}
    \caption{Distributions of $f$ and ratio of $W(f)$. The histograms show the distributions of $f$ in the training set (top), test set~A (middle) and test set~B (bottom). The last plot portrays $W(f+2)/W(f)$ evaluated with the same physical parameters used throughout the task, $\eta=4.01$ and $\lambda=1$. The markers represent even integer values of $f$,
    which enter the computation of $|\phi|^2$. In every histogram, we do not report weights below $10^{-4}$.}
    \label{fig:data_distr_reg_f}
\end{figure}

In Figs.~\ref{fig:data_distr_reg_kt} and~\ref{fig:data_distr_reg_f}, the first histogram corresponds to the distribution of the training set, while the second and third show the distributions of test sets~A and~B, respectively. The domains covered by the training set and test set~A are approximately the same, meaning that the generalization we require does not involve an extrapolation. Given this, it is easy to see why models that perform well during training and validation are able to generalize to different physical parameters. Despite having the same domain, there is an evident discrepancy in the distributions of the training set and test set~A, which is caused by the different values of $\mu$ that are used to generate the lattice configurations. This could be a possible source of generalization issues affecting some architectures. Note that while the distributions appear to be somewhat similar, there may still be additional differences in the correlations of these quantities, which could further impair generalization capabilities of networks. The domain that is covered by test set~B, however, is larger than the domain of the training set, arguably making the generalization to these data more demanding.

While the relationship between $\langle n\rangle$ and $k_t$ is linear, as shown in Eq.~\eqref{compl:n}, $\langle|\phi|^2\rangle$ is highly non-linear in $f$, as indicated by Eqs.~\eqref{compl:phi} and~\eqref{compl:W}. One might find it surprising that a CNN is able to learn such a complicated function and even generalize to other physical parameters. Alongside the observations on domain and distributions, we have to consider the ratio $W(f+2)/W(f)$ that enters in Eq.~\eqref{compl:phi}. As shown in Fig.~\ref{fig:data_distr_reg_f}, it can be effectively approximated by a linear function in the range where most of the distribution of the training set and test set~A is concentrated. This explains why even simple models can easily learn to predict $|\phi|^2$ on these data. The larger values of~$f$ that are represented in test set~B, however, lead to larger values of said ratio. At these values, the linear approximation that might be a good approximation on the training set and test set~A worsens, and so do the predictions of~$|\phi|^2$.

\subsection{Classification: detecting flux violations}

The algorithm presented in \cite{Gattringer:2013b} is designed to generate only closed worm configurations, which respect the flux conservation and allow to compute the observables $n$ and $\lvert \phi \rvert^2$. For the classification and regression tasks in sections~\ref{sec:Classification} and~\ref{sec:MultipleWorms}, we want to create field configurations where the flux conservation \eqref{compl:flux_conservation} is violated. In order to do this, we modify the algorithm of the previous subsection in the following manner: after equilibrium is reached via the original $l$ and $k$ alternate update, we start a new worm and save the configuration with one open worm. As the worm moves on the lattice, we replace the stored configuration with probability $1/L$, where $L$ is the current worm length, until the worm closes. One can easily check that this corresponds to selecting one of the open worm configurations with equal probability.

The dataset consists of closed worm configurations, labeled as class 0, and open worm configurations, labeled as class 1, each originating from two independent runs of the algorithm. Both classes are characterized by the same physical parameters, namely $\eta\in\{4.01,4.04,4.25\}$, $\lambda=1$, $\mu\in\{1,1.25,1.5\}$ and $N_t=N_x\in\{8,16,32,64\}$. The training set is generated on a particular subset, specifically the two combinations $(\eta,\mu)\in\{(4.01,1.5),(4.25,1)\}$ on the smallest lattice size, i.e. $8\times8$, with a total number of $N_{\mathrm{train}}=4000$ samples equally distributed among each class and parameter combination. The validation set has the same structure, the only difference being the number of samples of $N_{\mathrm{val}}=400$. The test set contains 100 instances per each class and parameter combination, summing up to $100\times2\times36=7200$ samples. The number of skipped configurations to avoid picking samples while the thermalization process is still ongoing is chosen as 2000. We use $100$  waiting sweeps between each measurement. The dataset created for this task and the next one share very similar characteristics. We address the analysis of only the third task in the following subsection, implying that the considerations we make there are also valid in this context.

\begin{figure}[ht]
    \centering
    \includegraphics{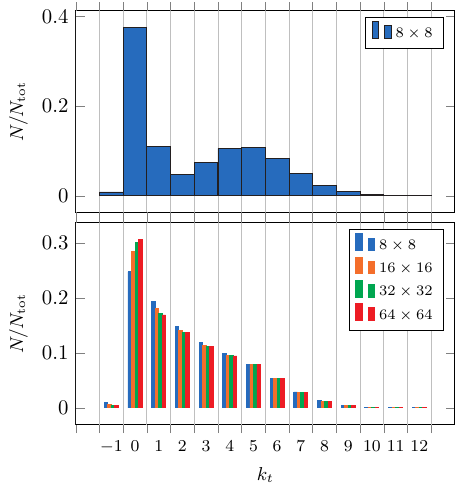}
    \caption{Distributions of the link field $k_t$. These two histograms feature the distributions of $k_t$ in the training set (top) and in the test set (bottom).}
    \label{fig:data_distr_reg2_kt}
\end{figure}

\begin{figure}[ht]
    \centering
    \includegraphics{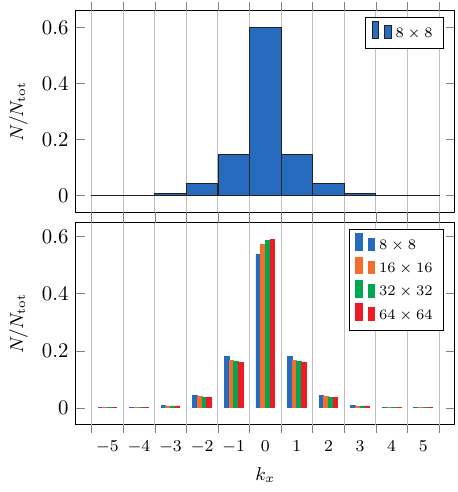}
    \caption{Distributions of the link field $k_x$. These two histograms feature the distributions of $k_x$ in the training set (top) and in the test set (bottom).}
    \label{fig:data_distr_reg2_kx}
\end{figure}

\subsection{Regression: counting flux violations}

The algorithm designed in the previous task is extended to account for multiple worms. After the first configuration with an open worm is saved as described in the last section, it becomes the starting configuration for the next worm to be drawn. We explicitly prohibit that a worm can cross heads and tails of previous worms, i.e. lattice sites where the flux is violated. By doing this, we ensure the absence of mathematical ambiguity in the definition of the Metropolis acceptance probability. As a consequence, three values for the flux are possible: 0, $+1$ and $-1$. The procedure is repeated until the required number of open worms is reached and the configuration is saved. Then the last configuration without open worms is restored, the established waiting sweeps are performed and another set of open worms is drawn.

The same set of physical parameters of the previous task is used with the addition of the number of worms $N_\mathrm{worms}\in\{0,1,\dots,10\}$. The subset of parameters for the training set is chosen as the combinations $(\eta,\mu)\in\{(4.01,1.5),(4.25,1)\}$ and $N_\mathrm{worms} \in\{0,5\}$, again on the smallest lattice size. The total amount of data used is $N_{\mathrm{train}}=20000$ when training only on 0 and five worms. The validation set consists as usual of a number of configurations such that $N_{\mathrm{val}} = N_{\mathrm{train}} / 10$ of the number of training samples. The test set contains again 100 samples per parameter combination, leading to a total of $100\times11\times36=39600$ instances. Initially skipped configurations and waiting sweeps are the same as in the previous task.

The two quantities necessary for the computation of the flux are the two integer fields $k_t$ and $k_x$, as suggested by Eq.~\eqref{compl:flux_conservation}. Their distributions are depicted, respectively, in Figs.~\ref{fig:data_distr_reg2_kt} and~\ref{fig:data_distr_reg2_kx}. We can draw the same conclusions as in the first task concerning the similarity of the domains and the difference in the distributions between training and test sets. We add that the choice of $(\eta,\mu)$ for the training set is made specifically to include the lower and higher values of $k_t$, in such a way that the domain covered is the same in the training and in the test set. This explains the two peaks in the $k_t$ training distribution in the top histogram in Fig.~\ref{fig:data_distr_reg2_kt}. Such behavior does not emerge in the case of $k_x$ because, unlike $k_t$, it is not coupled with the chemical potential.

\section{Additional Proofs} \label{app:Proofs}
\newtheorem{lemma}{Lemma}
\newtheorem*{remark}{Remark}
\newtheorem{theorem}{Theorem}

This appendix contains proofs referenced in section~\ref{sec:SymmetryProperties}. The idea of the first lemma is to show a simple, albeit arguably trivial example of a network's prediction being invariant under translations of the input even though the network contains a layer that breaks translational equivariance before the global pooling layer.

\begin{lemma}
    Given an $N \times N'$ feature map~$f$ and a $k \times k$ spatial average pooling layer~$P$ with stride \mbox{$s = k$}, \mbox{$k | N$} and \mbox{$k | N'$}, applying a global average pooling layer directly after the spatial average pooling layer is equivalent to applying only the global average pooling layer and omitting the spatial average pooling layer.
    \label{proofs:lemma:SpatAndGAP}
\end{lemma}

\begin{proof}
    We want to show that
     \begin{equation}
         \GAP ( P f(x) ) = \GAP ( f(y)). \label{proofs:SpatAndGAP}
    \end{equation}
    The global average pooling over an $N \times N'$ feature map~$f$ is given by
    \begin{equation}
         \GAP ( f(y) ) = \frac{1}{N N'} \sum_{y \in F} f(y), \label{proofs:GAP_def}
    \end{equation}
    with
    \begin{equation}
        f: F \subset \mathbb{Z}^2 \rightarrow \mathbb{R}.
    \end{equation}
    The spatial average pooling~$P$ can be interpreted as a special convolutional layer
    \begin{equation}
         P f(x) = [f \star \psi]_{s = k} (x) = \frac{1}{k^2} \sum_{\phi \in \Psi} f(k x + \phi), \label{proofs:spat_avg_pool_as_conv}
    \end{equation}
    using the filter $\psi(x) = 1/k^2$, where
    \begin{equation}
        \psi: \Psi \subset \mathbb{Z}^2 \rightarrow \mathbb{R}.
    \end{equation}
     The resulting feature map $f': F' \subset \mathbb{Z}^2 \rightarrow \mathbb{R}$ has the dimensions $N/k \times N'/k$. The validity of Eq.~\eqref{proofs:SpatAndGAP} can be seen by
     \begin{align}
         \nonumber \GAP ( P f(x) ) &= \frac{1}{(N / k) (N' / k)} \sum_{x \in F'} \frac{1}{k^2} \sum_{\phi \in \Psi} f(kx + \phi) \\
         \nonumber &= \frac{1}{N N'} \sum_{y \in F} f(y) \\
         &= \GAP ( f(y)).
     \end{align}
     The first step uses Eqs.~\eqref{proofs:GAP_def} and~\eqref{proofs:spat_avg_pool_as_conv}; note that the GAP is performed over the feature map~$f'$. The second equality holds for \mbox{$s=k$}, \mbox{$k | N$} and \mbox{$k | N'$}, and the last one utilizes again Eq.~\eqref{proofs:GAP_def} but for the feature map~$f$.
     \end{proof}

\begin{remark}
    Even though the spatial average pooling layer with $s > 1$ breaks translational equivariance under arbitrary translations, the result after the GAP is still invariant under translations.
\end{remark}

The following lemma shows that a convolutional layer that is directly followed by a global average pooling layer does not have the effect that one might expect a regular convolution to have. Loosely speaking, it does not effectively increase the network's depth because it collapses with the global pooling layer. This lemma should therefore also highlight the importance of the usage of an activation function between the last convolutional and the global average pooling layer.

\begin{lemma}
    Given an $M \times M'$ feature map~$f'$ and an $l \times l$ convolution $\psi'$: $\Psi' \subset \mathbb{Z}^2 \rightarrow \mathbb{R}$ with a stride of one and a single output channel, applying the convolution and then performing a global average pooling is equivalent to performing the global average, multiplying it by the sum of the convolution's weights and adding the bias~$b$. \label{proofs:lemma:ConvAndGAP}
\end{lemma}
\begin{proof}
    What we want to show is
    \begin{equation}
        \GAP ( [ f' \mspace{-3mu} \star \psi' ] (x) + b) = \GAP (f'(x)) \mspace{-6mu} \sum_{\phi' \in \Psi'} \mspace{-4mu} \psi'(\phi') + b. \label{proofs:ConvAndGAP}
    \end{equation}
    This can be seen by
    \begin{align}
        \GAP ( [ f' \mspace{-3mu} &\star \psi' ] (x) + b) \nonumber \\
        &= \frac{1}{M M'} \sum_{x \in F'} \mspace{-4mu} \left( \sum_{\phi' \in \Psi'} f'(x + \phi') \psi'(\phi') + b \right) \nonumber \\
        &= \frac{1}{M M'} \sum_{x \in F'} \sum_{\phi' \in \Psi'} f'(x + \phi') \psi'(\phi') + b \nonumber \\
        &= \sum_{\phi' \in \Psi'} \mspace{-4mu} \psi'(\phi') \frac{1}{M M'} \sum_{x \in F'} f'(x + \phi') + b \nonumber \\
        &= \sum_{\phi' \in \Psi'} \mspace{-4mu} \psi'(\phi') \mspace{5mu} \GAP (f'(x)) + b \nonumber \\
        &= \GAP (f'(x)) \mspace{-4mu} \sum_{\phi' \in \Psi'} \mspace{-4mu} \psi'(\phi') + b.
    \end{align}
    The first equality combines the definitions of the GAP, given by Eq.~\eqref{proofs:GAP_def}, and the convolution, given by Eq.~\eqref{layers:cross_correlation}. The second one utilizes the fact that the bias does not depend on the lattice site~$x$. The third step takes advantage of the fact that $\psi'$ does not depend on~$x$, and the fourth one makes use of the periodic boundary conditions and of Eq.~\eqref{proofs:GAP_def}. The last equality holds because the result of the GAP does not depend on~$\phi'$.
\end{proof}
\begin{remark}
    This is possible only without an activation function between the convolutional layer and the global average pooling. Also note the importance of periodic boundary conditions.
\end{remark}

The following theorem combines both lemmas and shows that a spatial average pooling layer, followed by a convolutional and a global average pooling layer, still leads to an output that is invariant under translations of the input if the strides and kernel sizes are chosen appropriately. It emphasizes once again the importance of an activation function before the global average pooling.

\begin{theorem}
    Given an $N \times N'$ feature map, a $k \times k$ spatial average pooling layer with stride $s = k$, $k | N$, $k | N'$ and an $l \times l$ convolution $\psi'$ with a stride of one and a single output channel, applying the spatial average pooling layer, then the convolution and then the global average pooling layer is equivalent to applying the global average pooling layer, multiplying the result by the sum of the convolution's weights and adding the bias~$b$. \label{proofs:theorem:SpatConvGAP}
\end{theorem}
\begin{proof}
    Combining lemmas~\ref{proofs:lemma:SpatAndGAP} and~\ref{proofs:lemma:ConvAndGAP} with $M = N/k$ and $M' = N'/k$ leads to the desired result.
\end{proof}
\begin{remark}
    The generalization to more than one feature map and multiple output channels is straightforward.
\end{remark}

\section{Partially occluded input} \label{app:PartiallyObscuredInput}

\begin{figure}[htbp]
	\centering
	\includegraphics{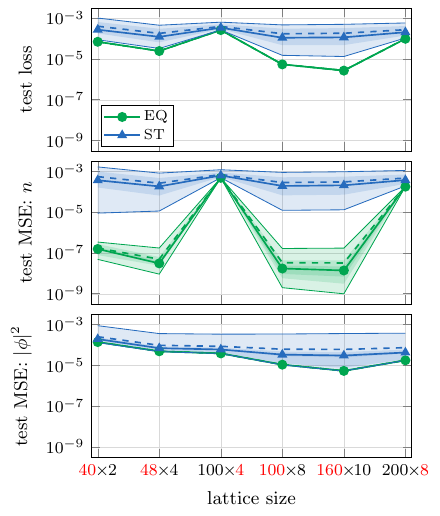}
	\caption{Test loss (top) and its two parts (middle and bottom) that come from each observable, corresponding to discarding $20 \%$ of the input data, on various lattice sizes. The networks are trained on a $60\times 4$ lattice. The dimension along which the input data are occluded is marked in red. For~$n$, the predictions become worse if the data are concealed in the spatial direction. For $\lvert \phi \rvert^2$, the predictions become worse if any data are hidden, but they are slightly worse if data are suppressed in the spatial direction. The reason for this lies in the nature of the observables.}
	\label{fig:test_loss_various_lattice_sizes_only_part_of_input}
\end{figure}

In this appendix, we analyze the worse performance of the \ST{} models on the \mbox{$100 \times 5$}~lattice from section~\ref{sec:Regression}, which is depicted in Fig.~\ref{fig:test_loss_various_lattice_sizes}. The source of the problem is that a model ignores part of the data at a strided operation if kernel and stride are not compatible with the data's shape that the layer in question receives. In the case of the \ST{} models, these operations are the strided spatial pooling layers. To examine this problem in more detail, we perform an experiment in which we hide a portion of the input data from the network for both architectures, \ST{} and \EQ{}: on every lattice size, the network is shown only a part of the input. This is done by discarding~\mbox{$20 \%$} of it either on the right or on the bottom of the lattice, so that the resulting restricted input still has a rectangular shape and thus a valid input size for the networks. This way, only~$80 \%$ of the input data are shown to the network, but the value of the observable still corresponds to the full lattice configuration. The input size of the \mbox{$50 \times 2$}~lattice becomes~\mbox{$40 \times 2$}, \mbox{$60 \times 4$} becomes~\mbox{$48 \times 4$}, \mbox{$100 \times 5$} becomes~\mbox{$100 \times 4$}, \mbox{$125 \times 8$} becomes~\mbox{$100 \times 8$}, and~\mbox{$200 \times 10$} becomes~\mbox{$160 \times 10$} and~\mbox{$200 \times 8$}, respectively. Note that the data are discarded four times in the temporal and twice in the spatial direction. The result of this experiment is shown in Fig.~\ref{fig:test_loss_various_lattice_sizes_only_part_of_input}. The overall test loss (top) shows two kinks, namely for the lattices, for which the input was restricted in the spatial direction. These kinks are also seen in the loss curve corresponding to~$\lvert \phi\rvert^2$ (bottom) but are much more pronounced for~$n$ (middle). In fact, if the data are discarded in the temporal direction, the quality of the prediction of~$n$ barely changes at all, as can be seen by comparing the middle plot in Fig.~\ref{fig:test_loss_various_lattice_sizes} to the middle plot in Fig.~\ref{fig:test_loss_various_lattice_sizes_only_part_of_input}. The reason for this is the way the configurations are generated, namely with the worm algorithm.

The nature of said algorithm is local, and each step involves adjacent points, giving rise to modifications of the field values that are contiguous on the lattice. This can be formally expressed by interpreting worms as paths on the torus-like space corresponding to the lattice with its periodic boundary conditions. For this task we deal with only closed worm configurations, so the paths are, in fact, loops. With this picture in mind, Eq.~\eqref{Regression:n_formula} represents an (averaged) winding number in the temporal dimension of the torus. Discarding data in this direction does not alter the winding number. On the other hand, if data are discarded along the other dimension, parts of the worm might be discarded as well, leading to a very high discrepancy from the true winding number.

The small kinks in $|\phi|^2$ can be explained by means of some additional remarks: first of all, in the range of $f_x$ in our dataset, the ratio $W(f_x+2)/W(f_x)$ is almost linear, so $|\phi|^2$ can be viewed as the average of $f_x$ in first approximation. The functions~$W$ and~$f_x$ are given by Eqs.~\eqref{compl:W} and~\eqref{f_x_definition}, respectively. The link variables $l_t$ and $l_x$ are not modified by the worm but by a standard Monte Carlo process; hence, their distribution is not biased in any direction, and the average over the truncated lattice has small deviations from the one over the whole lattice. Note that these deviations decrease as the lattice increases in size. Unlike the integer field $k_t$, $k_x$ is not coupled to the chemical potential; therefore, it is less likely for worms to wind around the spatial dimension than the temporal. This means that the average of $k_x$ is affected by a cut in the time dimension, but deviations from the true value do not occur often and decrease as the lattice size increases. Combining all these remarks finally yields the reason for those two kinks being at the same value of the restricted lattice size for both observables in Fig.~\ref{fig:test_loss_various_lattice_sizes_only_part_of_input} and the reason why they are less pronounced for $|\phi|^2$.

\bibliographystyle{elsarticle-num}
\bibliography{bibliography}

\end{document}